\documentclass[a4paper,english,cleveref,thm-restate]{lipics-v2021} 
\usepackage[utf8]{inputenc}

\hideLIPIcs  

\nolinenumbers
\usepackage{MnSymbol}

\usepackage[sort]{cite}
\usepackage{xspace}
\usepackage{csquotes}
\usepackage{nicefrac}
\usepackage{mathtools}
\usepackage{tikz}
\usepackage{tikz-cd}
\usetikzlibrary{calc}

\graphicspath{{./}{./figures/}}
\DeclareGraphicsExtensions{.pdf,.PDF}

\usepackage{subcaption}
\usepackage[normalem]{ulem}

\newcommand{\diam}{\ensuremath{\operatorname{diam}}}
\newcommand{\rad}{\ensuremath{\operatorname{rad}}}
\newcommand{\dist}{\ensuremath{\operatorname{dist}}}

\newcommand{\las}{\operatorname{ac}}

\newcommand{\sse}{\subseteq}
\newcommand{\pairs}{\mathcal{P}} 
\newcommand{\ennr}{N}
\newcommand{\ennb}{N'}
\newcommand{\sep}{separated\xspace}

\newcommand{\blowup}[2][k]{\ensuremath{{#1}\!\cdot\!#2}}
\newcommand{\blowupSeq}[1][forgottenparameter!!]{T_{#1}} 

\newcommand{\representation}{linear representation\xspace}

\DeclareMathOperator{\gone}{gone}

\crefname{equation}{Equation}{Equation}

\definecolor{deepcarrotorange}{rgb}{0.91, 0.41, 0.17}
\definecolor{ballblue}{rgb}{0.13, 0.67, 0.8}

\newcommand{\eshort}{short\xspace}   
\newcommand{\enormal}{near\xspace}
\newcommand{\eweird}{wide\xspace}

\newcommand{\TW}{\ensuremath{W}} 
\newcommand{\TWp}{\ensuremath{W'}} 
\newcommand{\TS}{\ensuremath{S}} 
\newcommand{\TSp}{\ensuremath{S'}} 
\newcommand{\TN}{\ensuremath{N}} 
\newcommand{\TNp}{\ensuremath{N'}} 
\newcommand{\largeN}{\ensuremath{\bar{n}}} 
 
\newcommand{\constantC}{\ensuremath{c}}
\newcommand{\IRS}{\ensuremath{I_{R,S}}} 
\newcommand{\IRSp}{\ensuremath{I_{R,S'}}} 

\crefformat{enumi}{#2(#1)#3}

\newtheorem*{theorem*}{Theorem}

\title{Flipping Non-Crossing Spanning Trees}

\authorrunning{H.~Bjerkevik, L.~Kleist, T.~Ueckerdt, and B.~Vogtenhuber}

\author{H\aa vard Bakke Bjerkevik}{University at Albany}{hbjerkevik@albany.edu}{https://orcid.org/0000-0001-9778-0354}{}

\author{Linda Kleist}{Universität Potsdam}{kleist@cs.uni-potsdam.de}{https://orcid.org/0000-0002-3786-916X}{}

\author{Torsten Ueckerdt}{Karlsruhe Institute of Technology}{torsten.ueckerdt@kit.edu}{https://orcid.org/0000-0002-0645-9715}{}

\author{Birgit Vogtenhuber}{Graz University of Technology}{birgit.vogtenhuber@tugraz.at}{https://orcid.org/0000-0002-7166-4467}{}



\acknowledgements{
    This work was initiated at the Dagstuhl Seminar 23221, in particular due to an interesting overview talk by Anna Lubiw.
    We thank the organizers and all the participants of the seminar for the fruitful atmosphere.
    H\aa vard Bakke Bjerkevik was funded by the Deutsche Forschungsgemeinschaft (DFG - German Research Foundation) - Project-ID 195170736 - TRR109.
}

\begin{document}

\maketitle

\begin{abstract}
    For a set $P$ of $n$ points in general position in the plane, the flip graph $\mathcal{F}(P)$ has a vertex for each non-crossing spanning tree on $P$ and an edge between any two spanning trees that can be transformed into each other by one edge flip, i.e., the deletion and addition of exactly one edge.
    The diameter $\diam(\mathcal{F}(P))$ of this flip graph is subject of intensive study. 
    For points $P$ in general position, it is between $\lfloor \nicefrac{3}{2}\cdot n \rfloor-5$ and $2n-4$, with no improvement for 25 years.
    For points $P$ in convex position, $\diam(\mathcal{F}(P))$ lies between $\lfloor \nicefrac{3}{2}\cdot n \rfloor -5$ and $\approx1.95n$, where the lower bound was conjectured to be tight up to an additive constant and the upper bound is a very recent breakthrough improvement over several previous bounds of the form $2n-o(n)$. 
    
    In this work, we provide new upper and lower bounds on the diameter of $\mathcal{F}(P)$ by mainly focusing on points $P$ in convex position. 
    We improve the lower bound even for this restricted case to $\diam(\mathcal{F}(P)) \geq \nicefrac{14}{9}\cdot n - \mathcal{O}(1)$.
    This disproves the conjectured upper bound of $\nicefrac{3}{2}\cdot n$ for convex position, while also improving the long-standing lower bound for point sets in general position.
    In particular, we provide pairs $T,T'$ of trees with flip distance $\dist(T,T') \geq \nicefrac{14}{9}\cdot n - \mathcal{O}(1)$; in these examples, both trees $T,T'$ have three boundary edges.
    We complement this by showing that if one of $T,T'$ has at most two boundary edges, then $\dist(T,T') \leq \nicefrac{3}{2}\cdot d < \nicefrac{3}{2}\cdot n$, where $d = |T-T'|$ is the number of edges in one tree that are not in the other.
    This bound is tight up to additive constants.

    Secondly, we significantly improve the upper bound on $\diam(\mathcal{F}(P))$ for $n$ points $P$ in convex position from $\approx 1.95n$ to $\nicefrac{5}{3}\cdot n - 3$.
    To prove both our lower and upper bound improvements, we introduce a new tool.
    Specifically, we convert the flip distance problem for given $T,T'$ to the problem of a largest acyclic subset in an associated \emph{conflict graph} $H(T,T')$.
    In fact, this method is powerful enough to give an equivalent formulation of the diameter of $\mathcal{F}(P)$ for points $P$ in convex position up to lower-order terms.
    As such, conflict graphs are likely the key to a complete resolution of this and possibly also other reconfiguration problems.
\end{abstract}

\newpage

\setcounter{page}{1}

\section{Introduction}
\label{sec:introduction}

Reconfiguration problems are important combinatorial problems with a high relevance in various settings and disciplines, e.g., robot motion planning, (multi agent) path finding, reconfiguration of data structures, sorting problems, string editing, in logistics, graph recoloring, token swapping, the Rubik's cube, or sliding puzzles, to name just a few.
Given a collection of configurations and a set of allowed reconfiguration moves, each transforming one configuration into another, we naturally obtain a (directed) graph on the space of all configurations.
When reconfiguration moves are reversible (then often called \emph{flips}), this graph is undirected and called a \emph{flip graph} $\mathcal{F}$.
For example, the flip graph of the Rubik's cube has more than $43\cdot 10^{18}$ vertices, each of degree~$27$.

A typical task is, for a pair $A,B$ of input configurations, to find a sequence of flips that transforms $A$ into $B$ -- preferably fast.
The distance between $A$ and $B$ in the flip graph $\mathcal{F}$ is the minimum number of required flips.
As computing (or even storing) the entire flip graph is usually impractical, one often resorts to the structure of~$\mathcal{F}$ to find a short flip sequence from $A$ to $B$.
However, even worst-case guarantees on the flip distance of $A$ and $B$ are mostly difficult to obtain.
It took 29 years and 35 CPU-years donated by Google to determine the largest flip sequence between any two Rubik's cubes, that is, to determine the diameter of the corresponding flip graph.
This elusive number is called God's number and equals $20$~\cite{Rokicki2010Rubik}.

Flip graphs are a versatile structure with many potential 
applications. For example, they are used to obtain Markov chains to sample random configurations, or for Gray codes and reverse search algorithms to generate all configurations. 
We give more related work in \cref{sec:relatedWork}, and refer to the survey articles~\cite{heuvel2003survey,nishimuraIntroReconfiguration} for even more examples and applications of reconfiguration problems. 

\begin{figure}[b!]
    \centering
    \includegraphics[page=2]{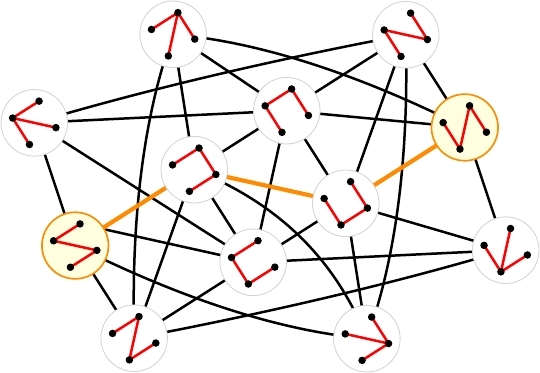}
    \caption{The flip graph $\mathcal{F}(P)$ on all non-crossing spanning trees on a set $P$ of $n=4$ points in convex position. A pair $T,T'$ with $\dist(T,T') = \diam(\mathcal{F}(P)) = 3$ is highlighted.}
    \label{fig:F4-example}
\end{figure}

A widely studied field concerns configuration of non-crossing straight-line graphs on a fixed point set in the plane.
In this setting, a flip is usually the exchange of one edge with another edge.
That is, two graphs  (i.e., configurations) $A,B$ are adjacent in the flip graph $\mathcal{F}$ if $|E(A) - E(B)| = |E(B)-E(A)| = 1$.
Classical examples are triangulations~\cite{Eppstein10,Eppstein10_socg07,rainbowJournal,rainbowSoCG18,HouleHNR05,HurtadoNU99,Lawson72,LubiwP15,pilz2014flip,pournin2014diameter,WagnerW22,WagnerW22_socg20,WagnerW22_soda20}, spanning trees~\cite{AichholzerAH02,aichholzer2022reconfiguration,AvisFukuda,bousquet2023noteJOURNAL,bousquet2024reconfigurationSoCG,Hernando,TreeTransition}, spanning paths~\cite{2023Aicholzer,akl2007planar,convexDiameter,KKR_paths,hamilton}, polygonizations~\cite{HernandoHH02}, and matchings~\cite{perfect-matchings,HouleHNR05,MilichMP21} on a fixed point set $P \subset \mathbb{R}^2$.
\enlargethispage{2ex}For an overview, see the survey article~\cite{bose2009flipsinplanar}.

Here, we study the flip graph of non-crossing spanning trees on a finite point set in the plane in general position. \cref{fig:F4-example} depicts the flip graph of non-crossing spanning trees on four points in convex position. 
Throughout, let $P$ denote a set of $n$ points in $\mathbb{R}^2$ with no three collinear points.
Consider a tree whose vertex-set is $P$ and whose edges are pairwise non-crossing straight-line segments.
Then a \emph{tree $T$ on $P$} is the edge-set of such a non-crossing spanning tree.
For instance, \cref{fig:Intro} shows some trees on a set $P$ in convex position.

\begin{figure}[htb]
    \centering
    \begin{subfigure}{.25\textwidth}
        \centering
        \includegraphics[page=8]{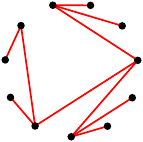}
        \caption{}
        \label{fig:IntroA}
    \end{subfigure}\hfil
    \begin{subfigure}{.25\textwidth}
        \centering
        \includegraphics[page=9]{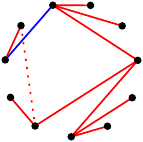}
        \caption{}
        \label{fig:IntroB}
    \end{subfigure}\hfil
    \begin{subfigure}{.25\textwidth}
        \centering
        \includegraphics[page=11]{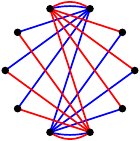}
        \caption{}
        \label{fig:IntroC}
    \end{subfigure}
    \caption{Some non-crossing trees on a point set in convex position.}
    \label{fig:Intro}
\end{figure}

Two trees $T$ and $T'$ on $P$ are related by a \emph{flip} if $T$ can be obtained from $T'$ by an exchange of one edge; i.e., there exist edges $e\in T$ and $e' \in T'$ such that $T'=T-e+e'$; see \cref{fig:IntroA,fig:IntroB} for an example. 
The \emph{flip graph $\mathcal F(P)$} of $P$ has a vertex for each tree on $P$ and an edge between any two trees that are related by a flip.
A path from $T$ to $T'$ in $\mathcal{F}(P)$ corresponds to a \emph{flip sequence} from $T$ to $T'$.
The length of a shortest flip sequence is the \emph{flip distance} of $T$ and $T'$, denoted by $\dist(T,T')$.
Finally, the \emph{diameter} of $\mathcal F(P)$ is the largest flip distance of any two trees on $P$, i.e., the smallest $D$ such that $\dist(T,T') \leq D$ for all $T,T'$.
In addition, the \emph{radius} of $\mathcal{F}(P)$ is $\rad(\mathcal{F}(P)) = \min_{T} \max_{T'} \dist(T,T')$.

The flip graph $\mathcal{F}(P)$ of trees on $P$ has been considered since 1996, when Avis and Fukuda~\cite{AvisFukuda} showed that any tree $T$ on $P$ can be flipped to any star $T'$ on $P$ whose central vertex lies on the boundary of the convex hull of $P$ in $|T-T'| \leq n-2$ steps.
This implies that $\mathcal F(P)$ has radius at most $n-2$ and hence diameter at most $2n-4$.
However, the exact radius and diameter of $\mathcal{F}(P)$ remain unknown to this day. 
Also, it is unclear how much the diameter varies between different point sets of same cardinality, or which point sets $P$ maximize the diameter of $\mathcal{F}(P)$ among all sets of $n$ points. 
In 1999, Hernando et al.~\cite[Theorem 3.5]{Hernando} provided a lower bound by constructing two trees $T,T'$ on $n$ points in convex position (for any $n \geq 4$) with flip distance $\dist(T,T') = \lfloor \nicefrac{3}{2}\cdot n\rfloor-5$; see \cref{fig:IntroC}. 
In this example, each edge in $T'-T$ intersects roughly half the edges of $T$.
Hence, every flip sequence from~$T$ to~$T'$ must flip away roughly $\nicefrac n2$ edges of $T$ before the first edge of $T' - T$ can be introduced, and thus $\dist(T,T')$ is at least roughly $\nicefrac 32 \cdot n$.
Yet, it remained open whether there is another pair of trees with larger flip distance.
As the only matching upper bound, we know that $\dist(T,T') \leq \lfloor\nicefrac{3}{2}\cdot n\rfloor-2$ in the special case that one of $T,T'$ is an $x$-monotone path~\cite{aichholzer2022reconfiguration}.

Note that the lower bound of $\lfloor \nicefrac 32 \cdot n \rfloor -5$ uses a point set in \emph{convex position}.
Interestingly, already this restricted setting is very challenging. 
The lower bound of $\lfloor \nicefrac 32 \cdot n \rfloor-5$ has not been improved for decades and the upper bound of $2n-4$ only was gradually improved in recent years.
In 2023, Bousquet et al.~\cite{bousquet2023noteJOURNAL} showed that $\dist(T,T') \leq 2n-\Omega(\sqrt n)$ for any two trees $T,T'$ on $n$ points in convex position and conjectured that $\nicefrac{3}{2}\cdot n$ flips always suffice.

\begin{conjecture}[Bousquet et al.~\cite{bousquet2023noteJOURNAL}]
    For any set $P$ of $n$ points in convex position, the flip graph $\mathcal{F}(P)$ has diameter at most $\nicefrac{3}{2}\cdot n$.
    \label{conj:convex-conjecture}
\end{conjecture}

\cref{conj:convex-conjecture} claims that every pair $T,T'$ of trees on a convex point set $P$ admits a flip sequence from~$T$ to~$T'$ of length at most $\nicefrac32 \cdot n$.
This is confirmed only for special cases, namely when one of $T,T'$ is a path~\cite{aichholzer2022reconfiguration} or a so-called \emph{separated caterpillar}~\cite{bousquet2023noteJOURNAL} (defined below).

It is also natural to compare the flip distance $\dist(T,T')$ of two trees with the trivial lower bound given by the number of edges in which $T$ and $T'$ differ, formally defined as $d = d(T,T') = |T-T'|$.
If $P$ is in convex position, it is easy to show that $d \leq \dist(T,T') \leq 2d-4$.
In 2022, Aichholzer et al.~\cite{aichholzer2022reconfiguration} showed that in fact $\dist(T,T') \leq 2d-\Omega(\log d)$.
Recently, Bousquet et al.~\cite{bousquet2024reconfigurationSoCG} broke the barrier of $2$ in the leading coefficient by showing that $\dist(T,T') \leq 1.96 d < 1.96 n$.
They also give a pair $T,T'$ with flip distance $\dist(T,T') \approx \nicefrac{5}{3}\cdot d$.
However, as their pair $T,T'$ has $d = |T-T'| \approx \nicefrac n2$, this is not a counterexample to \cref{conj:convex-conjecture}.

\subparagraph{Our contributions.}

We consider non-crossing trees on sets $P$ of $n$ points in convex position.
Our main results are significantly improved lower and upper bounds on the diameter of the corresponding flip graph $\mathcal{F}(P)$ in terms of $n$.
As all $n$-element convex point sets $P$ give the same flip graph $\mathcal{F}(P)$, let us denote it by $\mathcal{F}_n$ for brevity.
Recall that it is known that the diameter $\diam(\mathcal{F}_n)$ of $\mathcal{F}_n$ lies between roughly $1.5n$~\cite{Hernando} and $1.95n$~\cite{bousquet2024reconfigurationSoCG}.

We improve the upper bound to $\nicefrac{5}{3}\cdot n = 1.\overline{6}n$.

\begin{theorem}
    \label{thm:main-upper-bound}
    For any set $P$ of $n \geq 3$ points in convex position, the flip graph $\mathcal{F}(P)$ of non-crossing spanning trees on $P$ has diameter at most $\nicefrac 53 \cdot n - 3$.
    That is, $\diam(\mathcal{F}_n) \leq \nicefrac 53 \cdot n - 3$.
\end{theorem}

Secondly, we improve the known lower bound 
to roughly $\nicefrac{14}9 \cdot n = 1.\overline{5}n$.

\begin{restatable}{theorem}{mainlowerbound}\label{thm:main-lower-bound}
    There is a constant $\constantC$ such that for any $n \geq 1$, there are non-crossing trees $T,T'$ on $n$ points in convex position with $\dist(T,T')\geq \nicefrac{14}{9}\cdot n- \constantC$.
    That is, $\diam(\mathcal{F}_n) \geq \nicefrac{14}{9} \cdot n - \constantC$.
\end{restatable}

\cref{thm:main-lower-bound} is the first improvement over $\diam(\mathcal{F}_n) \geq \lfloor \nicefrac 32 \cdot n \rfloor - 5$, as given 25 years ago by the example of Hernando et al.~\cite{Hernando} depicted in \cref{fig:IntroC}.
Moreover, \cref{thm:main-lower-bound} disproves
\cref{conj:convex-conjecture} and also improves the lower bound on the largest diameter of $\mathcal{F}(P)$ among all point sets $P$ in general (not necessarily convex) position.

The trees $T,T'$ which we use to prove \cref{thm:main-lower-bound} have three boundary edges each, where a \emph{boundary edge} is an edge on the boundary of the convex hull of the underlying point set.
On the other hand, every non-crossing tree on a convex point set contains at least two boundary edges (provided $n \geq 3$), and trees with exactly two boundary edges are called \emph{separated caterpillars}.
Complementing \cref{thm:main-lower-bound}, we show that if at least one of $T,T'$ is a separated caterpillar, then their flip distance $\dist(T,T')$ is at most $\nicefrac{3}{2}\cdot d(T,T')$.
This improves on the recent upper bound of $\dist(T,T') \leq \nicefrac{3}{2}\cdot n$ for the same setting in \cite{bousquet2023noteJOURNAL}.
Further, the bound is tight up to an additive constant since the construction from~\cite{Hernando} in \cref{fig:IntroC} consists of two separated caterpillars.

\begin{restatable}{theorem}{caterpillar}\label{thm:caterpillar}
    Let $T,T'$ be non-crossing trees on $n\geq 3$ points in convex position.
    Let $T$ be a separated caterpillar and $d := |T-T'|$.
    Then $\dist(T,T') \leq \nicefrac{3}{2}\cdot d$.
    Moreover, there exists a flip sequence from $T$ to $T'$ of length at most $\nicefrac{3}{2}\cdot d$  in which no common edges are flipped.
\end{restatable}

Concerning sets $P$ of $n$ points in general (not necessarily convex) position, Aichholzer et al.~\cite[Open Problem~3]{aichholzer2022reconfiguration} asked for the radius of the flip graph $\mathcal{F}(P)$, in particular for a lower bound of the form $n-C$ for some small constant~$C$.
Avis and Fukuda~\cite{AvisFukuda} showed that the radius is at most $n-2$.
In fact, a matching lower bound is easily obtained.

\begin{theorem}
    \label{thm:main-radius-lower-bound}
    For any set $P$ of $n\geq 2$ points in general position, the flip graph $\mathcal{F}(P)$ of non-crossing trees on $P$ has radius at least (and thus exactly) $n-2$.
\end{theorem}
\begin{proof}
    Let $T$ be any tree on $P$, $v$ be a leaf of $T$, and $S_v$ be the star on $P$ with central vertex~$v$.
    Then $\dist(T,S_v) \geq d(T,S_v) = |T-S_v| = n-2$.
    As $T$ was arbitrary, the result follows.
\end{proof}

\subparagraph{Organization of the paper.}

We give an outline of our approach in \cref{sec:outline}; in particular we explain our strategy of reducing the task of determining the diameter of $\mathcal F_n$ to finding largest acyclic subsets of an associated conflict graph.
Our main tool is \cref{thm:alpha} (stated below).
In \cref{sec_fundamentals}, we define the conflict graphs and show how to derive \cref{thm:main-upper-bound} (up to a small additive constant) and \cref{thm:main-lower-bound} from \cref{thm:alpha}.
Then, \cref{sec_proof_alpha} is devoted to the proof of \cref{thm:alpha}.
In \cref{sec_improved_upper} we refine our tools to obtain an improved upper bound on $\dist(T,T')$ depending on $d(T,T') = |T-T'|$ and the number of boundary edges in $T \cap T'$, which also completes the proof of \cref{thm:main-upper-bound}.
In \cref{sec:caterpillars}, we study the case where one tree is a separated caterpillar and show \cref{thm:caterpillar}.
We conclude with a list of interesting open problems in \cref{sec:conclusions}.

\subsection{Related Work}
\label{sec:relatedWork}

First, let us mention further graph properties of the flip graphs $\mathcal F(P)$ of non-crossing trees on point set $P$ that have been investigated.
For $P$ in convex position, Hernando et al.~\cite{Hernando} showed that $\mathcal{F}(P)$ has radius $n-2$ and minimum degree $2n-4$, and that $\mathcal{F}(P)$ is Hamiltonian and $(2n-4)$-connected~\cite{Hernando}.
For point sets $P$ in general (not necessarily convex) position, Felsner et al.~\cite{rainbowJournal,rainbowSoCG18} showed that their flip graphs $\mathcal{F}(P)$ have so-called $r$-rainbow cycles for all $r=1,\dots, n-2$, which generalize Hamiltonian cycles.

\subparagraph{Restricted variants of flips for spanning trees.} 

Besides the general edge exchange flip (that we consider here), several more restricted flip operations have been investigated.
There is the \emph{compatible edge exchange} (where the exchanged edges are non-crossing), 
the \emph{rotation} (where the exchanged edges are adjacent), and the \emph{edge slide} (where the exchanged edges together with some third edge form an uncrossed triangle).
Nichols et al.~\cite{TreeTransition} provided a nice overview of the best known bounds for five studied flip types. 
Let us remark that for all five flip types, the best known lower bound in terms of $n$ (in the convex setting) corresponds to the general edge exchange.
Consequently, our \cref{thm:main-lower-bound} translates to all these settings.
In terms of $d = |T-T'|$, Bousquet et al.~\cite[arXiv version]{bousquet2024reconfigurationSoCG} gave a tight bound of $2d$ for sets of $n$ points in convex position and compatible edge exchanges as flip operation.
They also showed a lower bound of $\nicefrac 73 \cdot d$ for rotations as flip operation~\cite[arXiv version]{bousquet2024reconfigurationSoCG}.  
Their lower bound examples have $d \approx \nicefrac 23 \cdot n$ for compatible edge exchanges and $d \approx \nicefrac n2$ for rotations. 
For a variant with edge labels (which are transferred in edge exchanges), Hernando et al.~\cite{hernandoTreeLabeled} showed that the flip graph remains connected for any set $P$ in general position.

Lastly, let us mention reconfiguration of spanning trees in combinatorial (instead of geometric) settings, such as with leaf constraints~\cite{bousquet_et_al:LIPIcs.ESA.2020.24}, or degree and diameter constraints~\cite{bousquetTreesDegree}.

\subparagraph{Spanning paths.}

Much less is known when restricting $\mathcal{F}(P)$ only to the spanning paths on $P$.
In fact, the question of whether this subgraph $\mathcal{F}'(P)$ of $\mathcal{F}(P)$ is connected has been open for more than 16 years~\cite{akl2007planar,bose2009flipsinplanar}.
Akl et al.~\cite{akl2007planar} conjectured the answer to be positive, while confirming it if $P$ is in convex position.
In fact, Chang and Wu~\cite{convexDiameter} proved that for $n$ points $P$ in convex position, $\mathcal{F}'(P)$ has diameter $2n-5$ for $n=3,4$ and~$2n-6$ for all $n\geq5$.
It is also known that $\mathcal{F}'(P)$ is Hamiltonian~\cite{hamilton} and has chromatic number $\chi(\mathcal{F}'(P)) = n$~\cite{chromatic}.

For $P$ in general position, $\mathcal{F}'(P)$ is known to be connected for so-called generalized double circles~\cite{2023Aicholzer}, and its diameter is at least $2n-4$ if $P$ is a wheel of size $n$~\cite{2023Aicholzer}.
Kleist, Kramer, and Rieck~\cite{KKR_paths} showed that so-called \emph{suffix-independent paths} induce a large connected subgraph in $\mathcal{F}'(P)$, and confirmed connectivity of $\mathcal{F}'(P)$ if $P$ has at most two convex~layers.

\subparagraph{Triangulations.}

For (non-crossing) inner triangulations on a set $P$ of $n$ points in general position, a flip replaces a diagonal of a convex quadrilateral spanned by two adjacent inner faces by the other diagonal.
When points in $P$ are in convex position, the corresponding flip graph $\mathcal{T}(P)$ is the $1$-skeleton of the $(n-3)$-dimensional associahedron. 
In fact, the vertices of $\mathcal{T}(P)$ are in bijection with binary trees and the flip operation with rotations of these trees.

The diameter of $\mathcal{T}(P)$ is known to be in $\Omega(n^2)$ for $P$ in general position~\cite{HurtadoNU99}, and at most $2n-10$ (for $n \geq 9$) for $P$ in convex position~\cite{sleator1986rotation}, where the latter is in fact tight~\cite{pournin2014diameter}.

Computing the flip distance of two triangulations on $P$ is known to be \textsc{NP}-complete~\cite{LubiwP15,pilz2014flip}, also in the more general setting of graph associahedra~\cite{graphAssoHardness}. 
Many further properties of the associahedron have been investigated, such as geometric realizations~\cite{CSZ15}, Hamiltonicity~\cite{lucas1987}, rainbow cycles~\cite{rainbowJournal,rainbowSoCG18}, and expansion and mixing properties~\cite{EF23}.

\section{Outline of Our Approach}
\label{sec:outline}

Let us outline our approach to tackle the diameter $\diam(\mathcal{F}_n)$ of the flip graph $\mathcal{F}_n$ for $n$ points in convex position.
Together, \cref{thm:main-upper-bound,thm:main-lower-bound} state that
\[
    \nicefrac{14}9 \cdot n - \mathcal{O}(1) \leq \diam(\mathcal{F}_n) < \nicefrac 53 \cdot (n-1) = \nicefrac{15}9 \cdot (n-1),
\]
narrowing the gap from roughly $\nicefrac 12 \cdot n$ to only $\nicefrac 19 \cdot n$.
We obtain both the upper and the lower bound by transferring the question for the diameter of the flip graph into a more approachable question about largest acyclic subsets in certain conflict graphs.
Our corresponding result is stated in \cref{thm:alpha} below.
While we defer the precise definitions to \cref{sec_fundamentals}, let us provide here some background needed to understand \cref{thm:alpha} and explain how \cref{thm:alpha} could be used to determine $\diam(\mathcal{F}_n)$ exactly up to lower-order terms.

Given a pair $T,T'$ of trees on a set $P$ of $n$ points in convex position, we define a canonical bijection between the edges in $T$ and the edges in $T'$, formalized as a set $\pairs$ of pairs $(e,e')$ with $e \in T$ and $e' \in T'$.
So, each $e \in T$ has a unique partner $e' \in T'$, and vice versa.
We then restrict our attention to flip sequences from $T$ to $T'$ that respect this bijection in the sense that every $e\in T$ is flipped to its partner $e' \in T'$ in at most two steps.
That is, either $e$ is flipped to $e'$ directly (a \emph{direct} flip), or $e$ is flipped to $e'$ in two steps via one intermediate boundary edge (an \emph{indirect} flip).
The length of a flip sequence of this form is then $\#$direct flips + $2\cdot \#$indirect flips, and our task is to minimize the number of indirect flips.
This gives an upper bound on the flip distance of any pair $T,T'$ of trees, and hence an upper bound on $\diam(\mathcal{F}_n)$.

Of course, this approach is a massive restriction on the allowed flip sequences, while we must (for a particularly chosen pair $T,T'$) argue about \emph{every} flip sequence in order to get a lower bound on $\diam(\mathcal{F}_n)$.
But surprisingly, it turns out that to determine $\diam(\mathcal{F}_n)$ asymptotically, it suffices to consider these restricted flip sequences.
The crucial idea is to ``blow up'' a small example pair of trees so that \emph{every} flip sequence for the resulting blown-up pair of trees ``acts like'' our restricted flip sequences for ``most of the edges''.
Then, if the small example pair of trees requires many flips in every restricted flip sequence, then the blown-up pair of trees requires many flips in \emph{every} flip sequence, which hence gives a lower bound on $\diam(\mathcal{F}_n)$.

In order to analyse the number of required flips in any restricted flip sequence for a given pair $T,T'$, we associate a directed \emph{conflict graph} $H = H(T,T')$ whose vertices correspond to a subset of the pairs in $\pairs$.
A directed edge $(e_1,e'_1) \to (e_2,e'_2)$ in $H$ expresses that the direct flip $e_2\to e'_2$ cannot occur before the direct flip $e_1 \to e'_1$, as otherwise it would create a cycle or a crossing.
Let $\las(H)$ denote the size of a largest subset of $V(H)$ that induces an acyclic subgraph.
We then construct a flip sequence from $T$ to $T'$ with $\las(H)$ direct flips.
So, if $\las(H)$ is large, then $\dist(T,T')$ is small.
On the other hand, if $\las(H)$ is small, we can derive a good asymptotic lower bound on $\diam(\mathcal{F}_n)$.
The precise statements go as follows.

\begin{theorem}
    \label{thm:alpha}
    Let $T,T'$ be two non-crossing trees on linearly ordered points $p_1,\ldots,p_n$ with 
    corresponding conflict graph $H = H(T,T')$.
    \begin{romanenumerate}
        \item If $V(H)$ is non-empty, then $\dist(T,T')\leq \max\left\{\frac32,2-\frac{\las(H)}{|V(H)|}\right\} (n-1)$.\\
            If $V(H)$ is empty, then $\dist(T,T')\leq \frac32 (n-1)$.
            \label{item:alpha-upper-bound}
        
	\item If $V(H)$ is non-empty, then there is a constant $\constantC$ {depending only on $T$ and $T'$} such that for all $\largeN\geq 1$, we have $\diam(\mathcal{F}_{\largeN}) \geq \left(2-\frac{\las(H)}{|V(H)|}\right)\largeN - \constantC$. 
            \label{item:alpha-lower-bound}
    \end{romanenumerate}
\end{theorem}

\cref{thm:alpha} implies that there exists a constant $\gamma \in [\nicefrac 32,2]$ such that
\begin{equation}
    \lim_{n\to \infty} \frac{\diam(\mathcal{F}_n)}{n} = \gamma \coloneqq \sup_H \left(2-\frac{\las(H)}{|V(H)|}\right),
    \label{eq:gamma}
\end{equation}
where the supremum is taken over all non-empty conflict graphs $H$ arising from pairs of non-crossing trees.
Indeed, we have the stronger statement
\begin{equation*}
    \gamma n - o(n) \leq \diam(\mathcal{F}_n) \leq \gamma (n-1),
\end{equation*}
where the second inequality follows from \cref{thm:alpha} (i) (this assumes $\gamma\geq \nicefrac 32$, which can be proved using the example in \cref{fig:IntroC} -- we omit the details, since we show a better bound in \cref{lem_example_14_9}); \cref{thm:alpha} (ii) gives us a family of bounds $\gamma' n - O(1) \leq \diam(\mathcal{F}_n)$ for a set of $\gamma'$ whose supremum is $\gamma$, which implies the first inequality.

In the light of~\eqref{eq:gamma}, the task of bounding $\diam(\mathcal{F}_n)$ looks quite different.
As evidenced by our own results, this is a major simplification:
With \cref{thm:alpha}\eqref{item:alpha-lower-bound} at hand, we can prove a lower bound on $\diam(\mathcal{F}_n)$ for all $n$ quite easily.
It is enough to construct a single example of two trees $T,T'$ on a point set $P$ in convex position, and to compute a largest acyclic subset of the corresponding conflict graph $H$.
In fact, all we do to prove \cref{thm:main-lower-bound} is exhibit an example of two trees on $13$ points, compute their conflict graph $H$ on $9$ vertices, and write a two-line proof that $\las(H)\leq 4$; see \cref{lem_example_14_9}.
To further improve on our lower bound (if possible), one simply needs to do the same with a better example pair of trees.

And with \cref{thm:alpha}\eqref{item:alpha-upper-bound} at hand, we also prove our upper bound in \cref{thm:main-upper-bound} through conflict graphs.
We divide the vertices of the conflict graph $H$ arising from an arbitrary pair $T,T'$ of trees into three sets, and show that each set induces an acyclic subgraph of~$H$; see \cref{lem:ABC-acyclic}.
Also, showing this acyclicity requires only a short argument.

Thus, \cref{thm:alpha} allows succinct proofs of upper and lower bounds on $\diam(\mathcal{F}_n)$.
Any improved lower bounds on $\frac{\las(H)}{|V(H)|}$, or examples of conflict graphs $H$ with $\frac{\las(H)}{|V(H)|} < \nicefrac{4}{9}$, will give improved bounds on $\diam(\mathcal{F}_n)$.
This is a promising avenue towards determining the exact value of~$\gamma$.
Moreover, the conflict graph might be useful in other reconfiguration problems.

\section{Conflict Graphs, Acyclic Subsets, and the Diameter of $\mathcal{F}_n$}
\label{sec_fundamentals}

Throughout this section, let $P$ be a set of $n$ points in the plane in convex position.
As the flip graph $\mathcal{F}(P) = \mathcal{F}_n$ only depends on $n$, we can imagine the points in $P$ to lie equidistant on a circle, circularly labeled as $p_1,\dots, p_n$.
Given a tree $T$ with vertex-set $P$, we can represent $T$ as a straight-line drawing on $P$.
If this drawing has no crossing edges, then $T$ is non-crossing and we simply call $T$ a \emph{tree on $P$}.
For convenience, we treat each tree $T$ as its set of edges.
Let $\mathcal{T}_n$ denote the set of all trees on $P$.
If for two trees $T,T'$ on $P$ we have $|T-T'| = 1$, then $T$ and $T'$ are related by a \emph{flip}.
The \emph{flip graph} $\mathcal{F}_n$ has vertex-set $\mathcal{T}_n$ and an edge for any two trees on $P$ that are related by a flip.
Given two trees $T, T' \in \mathcal T_n$, the \emph{flip distance} $\dist(T,T')$ is the length of a shortest path from~$T$ to $T'$ in $\mathcal{F}_n$.

Consider two fixed trees $T,T'$ on $P$. 
We work with a \emph{\representation} as illustrated in \cref{fig:cut-openB} with an example.
Intuitively speaking, we cut open the circle between $p_1$ and $p_n$ and unfold the circle into a horizontal line segment, usually called the \emph{spine}.
Each edge in $T \cup T'$ was a straight-line chord of the circle and can now be thought of as a semi-circle above or below the spine.
For better readability, we usually put the edges of $T$ above and the edges of $T'$ below the spine, see again \cref{fig:cut-open}.

\begin{figure}[ht]
    \centering
    \begin{subfigure}{.3\textwidth}
        \centering
        \includegraphics[page=1]{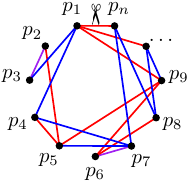}
        \caption{}
        \label{fig:cut-openA}
    \end{subfigure}\hfil
    \begin{subfigure}{.4\textwidth}
        \centering
        \includegraphics[page=2]{Intro}
        \caption{}
        \label{fig:cut-openB}
    \end{subfigure}
    \caption{
        (a) Two non-crossing trees $T,T'$ on a circularly labeled point set in convex position and (b) its \representation with $T$ above and $T'$ below the horizontal spine.
    }
    \label{fig:cut-open}
\end{figure}

By the linear order $p_1,\ldots,p_n$ we have also a natural notion of the length of an edge.
That is, if edge $e$ has endpoints $p_i$ and $p_j$, then the \emph{length} of $e$ is $|i-j|$.
Moreover, we say that an edge $e$ with endpoints $p_i$ and $p_j$, $i<j$, \emph{covers} a vertex $p_k$ if $i \leq k \leq j$.
In particular, each edge covers both of its endpoints.
An edge $e$ \emph{covers} an edge~$f$ if $e$ covers both endpoints of~$f$.
If no edge $e \in T - f$ covers the edge $f \in T$, then we say that $f$ is an \emph{uncovered edge} in $T$.

The linear order of the $n$ points of $P$ also defines $n-1$ \emph{gaps} $g_1,\ldots, g_{n-1}$, where gap $g_i$ simply is the (open) segment along the spine with endpoints $p_i$ and $p_{i+1}$.
To introduce a few crucial properties, let us consider any set~$E$ of non-crossing edges on $p_1,\ldots,p_n$ (not necessarily forming a tree).
For each gap $g$ that is covered by at least one edge of $E$, let $\rho_E(g)$ be the shortest edge of $E$ covering $g$. If no such edge exists then $\rho_E(g)$ is undefined.
If $E$ is a tree on $P$, then all gaps are covered by $E$, while $|E| = |P|-1$ is exactly the number of gaps.
The following lemma shows that $\rho_E$ forms a bijection between gaps and edges in $E$ if and only if $E$ is a tree on $P$.

\begin{lemma}
    \label{lem_gap_edge_bij}
    Let $E$ be a set of non-crossing edges on a linearly labeled point set $P$.
    Then $\rho_E$ defines a bijection from the set of gaps to $E$ if and only if $E$ forms a tree on $P$.
\end{lemma}
\begin{proof}
    Suppose first that $\rho_E$ is a bijection.
    We argue by induction on $|E|=|P|-1$ that $E$ is a tree.
    If $|E| \le 1$ then the statement is trivially true.
    If $|E| \ge 2$, let $e \in E$ be an edge that is not covered by any other edge in $E$, and let $g_i = \rho_E^{-1}(e)$ be the corresponding gap.
    Since $e$ is the only edge that covers gap $g_i$, it induces a partition $E = \{e\}\sqcup L\sqcup R$, where $L$ and $R$ contain the edges to the left and right of $g_i$, respectively.
    Restricting $\rho_E$ to these partitions gives bijections $\{g_1,\dots,g_{i-1}\}\to L$ and $\{g_{i+1},\dots,g_{n-1}\}\to R$. By induction, $L$ and $R$ are trees on $\{p_1,\dots,p_{i}\}$ and $\{p_{i+1},\dots,p_{n}\}$. Further, as $e$ covers $g_i$, we have $e=p_kp_\ell$ with $k \le i$ and $\ell \ge i+1$. 
    Hence $e$ connects the two trees $L$ and $R$, showing the $E$ is a tree on $P$.   

    Now, suppose that $E$ forms a tree. 
    Clearly, each gap is covered, so $\rho_E(g)$ is defined for each gap $g$.
    If $\rho_E$ is not injective, then there exist two gaps $g$, $g'$ and an edge $e$ such that $e = \rho_E(g) = \rho_E(g')$.
    Then the points of $P$ between $g$ and $g'$ are separated from the remaining points because any edge connecting them would be a shorter edge covering $g$ or $g'$; a contraction to the fact that $E$ forms a tree.
    Because both the number of edges and gaps are $n-1$, injectivity implies that $\rho_E$ is bijective.
\end{proof}

For a non-crossing tree $T$ on a linearly ordered point set $p_1,\ldots, p_n$, we define $e_i:=\rho_T(g_i)$ and categorize the edges of $T$ into three types, depending on how many endpoints of an edge~$e_i$ are also endpoints of its corresponding gap $g_i$.
For each $i \in [n-1]$, we say that the edge $e_i = \{u,v\}$ of $T$ is a  
\begin{itemize}
    \item \emph{\eshort{} edge} if $\{u,v\} = \{p_i,p_{i+1}\}$,
    \item \emph{\enormal{} edge} if $|\{u,v\} \cap \{p_i,p_{i+1}\}| = 1$, and
    \item \emph{\eweird{} edge} if $\{u,v\} \cap \{p_i,p_{i+1}\} = \emptyset$.
\end{itemize}

The set of all \eshort{}, \enormal{}, and \eweird{} edges of $T$ is denoted by $\TS$, $\TN$, and $\TW$, respectively.
Note that the \eshort{} edges of $T$ are the boundary edges of $T$ different from $p_np_1$, or in other words, the edges of length~$1$. 

Symmetrically, for a tree $T'$, we denote the edge corresponding to gap $g_i$ by $e'_i$ and the sets of 
all \eshort{}, \enormal{}, and \eweird{} edges of $T'$ by $\TSp$, $\TNp$, and $\TWp$, respectively.

\subparagraph{Pairing.}

Given $T,T'$ and a \representation, we define $\pairs = \{ (e_i,e'_i) \mid i = 1,\ldots,n-1 \}$ to be the natural pairing of the edges in $T$ with those in $T'$ according to their corresponding gap.
That is, $(e,e') \in \pairs$ for $e \in T$ and $e' \in T'$ if and only if $\rho^{-1}_T(e) = \rho^{-1}_{T'}(e')$.
Note that $e_i$ and $e'_i$ might coincide, i.e., $e_i = e'_i$;
in particular, this happens if $e_i$ is a \eshort{} edge in $T \cap T'$.

Next we partition the set $\pairs$ of edge pairs as follows:
\begin{itemize}
    \item $\pairs_= = \{ (e,e') \in \pairs \mid e = e'\}$,
	\item $\pairs_N = \{ (e,e') \in \pairs \mid e \neq e' \text{ and } e \in \TN \text{ and } e' \in \TNp\}$, and
	\item $\pairs_R = \pairs - (\pairs_= \cup \pairs_N)$.
\end{itemize}

Clearly, $|\pairs_=| + |\pairs_N| + |\pairs_R| = |\pairs| = n-1$. 
As it turns out, we will spend no flips on pairs in~$\pairs_=$ and it will be enough to spend in total at most $\frac32|\pairs_R| + |\pairs_=|$ flips on pairs in $\pairs_R$.
The more difficult part will be the pairs in $\pairs_N$, namely, the \enormal{}-\enormal{} pairs. The aim is to find a large subset of $\pairs_N$ which only needs one flip per edge.

\subparagraph{Conflict graph.}

We want to find a large set of \enormal{}-\enormal{} pairs that can be flipped directly.
However, two \enormal{}-\enormal{} pairs $(e_i,e'_i)$ and $(e_j,e'_j)$ could be so interlocked that it is impossible to have both as direct flips in any flip sequence from $T$ to $T'$.
This is for example the case if $e_i$ crosses $e'_j$ and $e_j$ crosses $e'_i$.
To capture all these dependencies we define a directed auxiliary graph which we call the \emph{conflict graph} $H$ of $T,T'$. 
The vertex set of $H$ are the gaps corresponding to the \enormal{}-\enormal{} pairs and a directed edge $\overrightarrow{g_i g_j}$ in $H$ implies the presence of $e_i$ prevents a direct flip from $e_j$ to $e'_j$.

Let $I_=,I_R,I_N$ denote the subsets of gaps corresponding to $\pairs_=,\pairs_R,\pairs_N$, respectively.

\begin{definition}[Conflict graph]{\ \\}
    \label{def_conflict_graph}
    The \emph{conflict graph} $H=H(T,T')$ is the directed graph defined by 
    \begin{itemize}
        \item $V(H) := I_N$; i.e., the vertices  are the gaps corresponding to \enormal{}-\enormal{} pairs, and
        \item there is a directed edge in $E(H)$ from $g_i$ to $g_j$, denoted $\overrightarrow{g_i g_j}$, if
        \begin{description}
            \item[type~1:] $e_i$ crosses $e'_j$, or
            \item[type~2:] $e'_j$ covers $e_i$ and $e_i$ covers $g_j$, or
            \item[type~3:] $e_i$ covers $e'_j$ and $e'_j$ covers $g_i$.
        \end{description}
    \end{itemize}
\end{definition}

\cref{fig:conflict-edge} illustrates the three types of edges in $H$.
\cref{fig:conflict-graph} depicts a full example of a \representation of a pair $T,T'$ and the corresponding conflict graph $H(T,T')$.
Observe that $H(T',T)$ is obtained from $H(T,T')$ by reversing the directions of all edges.

\begin{figure}[ht]
    \centering
    \begin{subfigure}{.24\textwidth}
        \centering
        \includegraphics[page=18]{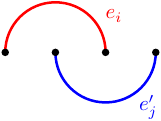}
        \caption{$\overrightarrow{g_i g_j}$ of type~1}
        \label{fig:conflict-edgeA}
    \end{subfigure}\hfill
    \begin{subfigure}{.24\textwidth}
        \centering
        \includegraphics[page=19]{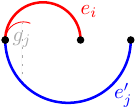}
        \caption{$\overrightarrow{g_i g_j}$ of type~2}
        \label{fig:conflict-edgeB}
    \end{subfigure}\hfill
    \begin{subfigure}{.24\textwidth}
        \centering
        \includegraphics[page=20]{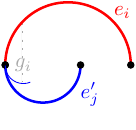}
        \caption{$\overrightarrow{g_i g_j}$ of type~3}
        \label{fig:conflict-edgeC}
    \end{subfigure}\hfill
    \begin{subfigure}{.24\textwidth}
        \centering
        \includegraphics[page=27]{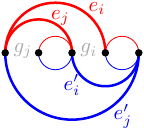}
        \caption{$\overrightarrow{g_i g_j}$ of type~2}
        \label{fig:conflict-edgeD}
    \end{subfigure}
    \caption{
        Examples of directed edges in the conflict graph: (a) type~1 (b) type~2 (c) type~3.
        Mirroring the examples in (a)-(c) horizontally gives a complete list of all possibilities.
        (d) Example of a conflict of type~2: the direct flip $e_j \to e_j'$ in $T$ (above, red) does not yield a tree.
    } 
    \label{fig:conflict-edge}
\end{figure}

\begin{figure}[ht]
    \centering
    \begin{subfigure}{.32\textwidth}
        \centering
        \includegraphics[page=22]{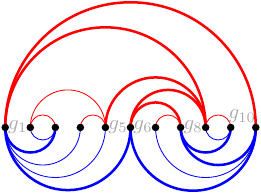}
        \caption{}
        \label{fig:conflict-graphA}
    \end{subfigure}\hfil
    \begin{subfigure}{.2\textwidth}
        \centering
        \includegraphics[page=23]{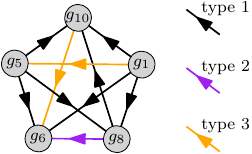}
        \caption{}
        \label{fig:conflict-graphB}
    \end{subfigure}
    \caption{
        (a) The trees $T,T'$ from \cref{fig:cut-open} with pairs in $\pairs_N$ (drawn with thick curves) and (b) their conflict graph~$H$.
    } 
    \label{fig:conflict-graph}
\end{figure}

A direct flip $e_j \to e'_j$ for a \enormal{}-\enormal{} pair $(e_j,e'_j)$ may be invalid for two reasons, either because the introduced edge crosses an existing edge $e_i$ (this corresponds in $H$ to an edge $\overrightarrow{g_ig_j}$ of type~1) or because the new graph is not a tree (this is captured by incoming edges at $g_j$ in $H$ of type~2 and 3).
In fact, we claim (and prove later) that a \enormal{}-\enormal{} edge pair $(e_j,e'_j)$ admits a direct flip $e_j \to e'_j$ (after flipping all pairs of $\pairs_= \cup \pairs_R$ to the boundary) if and only if the corresponding gap $g_j$ has no incoming edge in $H$.

In the remainder of this section, we show how \cref{thm:alpha} can be used to prove $$\nicefrac{14}{9} \cdot n - \mathcal{O}(1) \leq \diam(\mathcal{F}_n) \leq \nicefrac{5}{3}\cdot(n-1).$$

\subsection{Upper Bound on the Flip Distance via \cref{thm:alpha}\eqref{item:alpha-upper-bound}}
\label{subsec_upper}

In this subsection, we assume that \cref{thm:alpha}\eqref{item:alpha-upper-bound} holds, and show how to derive \cref{thm:main-upper-bound} from it.
That is, we prove an upper bound of $\nicefrac 53\cdot(n-1)$ on the flip distance of two non-crossing trees $T$ and $T'$ on $n$ points by finding a large acyclic subset of the corresponding conflict graph $H$.
By \cref{thm:alpha}\eqref{item:alpha-upper-bound} we have $\dist(T,T') \leq \max\left\{\frac32,2-\frac{\las(H)}{|V(H)|}\right\} (n-1)$, so we seek to prove that $\las(H) \geq \nicefrac 13 \cdot |V(H)|$ whenever $H$ is non-empty.

Recall that a \enormal{} edge is incident to exactly one vertex at its corresponding gap.
We can think of a \enormal{} edge $e_i$ (or $e'_i$) to ``start'' at gap $i$ (either at $p_i$ or $p_{i+1}$) and then ``go'' either left or right.
Clearly, each \enormal{}-\enormal{} pair $(e,e')$ starts at the same gap.
Moreover, observe that $e$ and~$e'$ go in the same direction if and only if $e$ and~$e'$ are adjacent, i.e., have a common endpoint (which is then necessarily at the gap).

In order to prove a lower bound on $\las(H)$, and hence an upper bound on $\dist(T,T')$, let us inspect the gaps more closely.
We partition the set $I_N$ of all \enormal{}-\enormal{} gaps into three subsets, distinguishing for each gap with a corresponding \enormal{}-\enormal{} pair, whether these two edges are adjacent, and (in case they are) which edge is longer.
For each gap $g_i \in I_N$ with \enormal{}-\enormal{} pair $(e_i,e'_i) \in \pairs_N$, we say that
\begin{itemize}
    \item $g_i$ is \emph{above} if $e_i$ and $e'_i$ are adjacent and $e_i$ is longer than $e'_i$,
    \item $g_i$ is \emph{below} if $e_i$ and $e'_i$ are adjacent and $e_i$ is shorter than $e'_i$, and
    \item $g_i$ is \emph{crossing} if $e_i$ and $e'_i$ are not adjacent.
\end{itemize}
We denote the set of all above (respectively below, crossing) gaps in $I_N$ by $A$ (respectively $B$, $C$).
By definition, $A$, $B$, $C$ are pairwise disjoint, and hence $|A| + |B| + |C| = |I_N| \leq n-1$.

\begin{lemma}
    \label{lem:ABC-acyclic}
    Each of $A$, $B$, $C$ is an acyclic subset of $H$.
    In particular, $\las(H) \geq \nicefrac13 \cdot|V(H)|$.
\end{lemma}
\begin{proof}
    To prove that $Y \in \{A,B,C\}$ is acyclic, we show that there is some gap $g^* \in Y$ without incoming edges in $H[Y]$.
    By removing $g^*$ from $Y$ and repeating the argument, it follows that $Y$ is acyclic.
    We separately consider the three possible choices of $Y$.

    \begin{description}
	\item[Case $Y = A$:]
            Consider a gap $g_j \in A$ such that the length of $e_j$ is minimal.
            We claim that the gap $g_j$ has no incoming edge in $H[A]$. 
            Without loss of generality, we assume that $e_j$ and $e'_j$ share their left endpoint as illustrated in \cref{fig:ABC-acyclicA}; otherwise consider the mirror image.
            We show that $g_j$ has no incoming edge in $H[A]$ of any of the three types~1, 2, and~3.
            By the choice of $g_j$, $e_j$ does not cover any edge $e_i$ with $g_i \in A$; otherwise $e_i$ would be shorter than $e_j$.
            This excludes incoming edges of type~1 and~2.
            Further, for any $g_i \in A$ with $e_i$ covering $e_j$ and $e'_j$, then the gap $g_i$ cannot be covered by $e'_j$.
            This excludes incoming edges of type~3.
		
	\item[Case $Y = B$:]
	    This is symmetric to the previous case by exchanging the roles of $T$ and $T'$.
            In particular, any $g_j \in B$ such that the length of $e'_j$ is minimal has no incoming edges in $H[B]$.

	\item[Case $Y = C$:]
            Observe from \cref{fig:conflict-edge} that all directed edges in $H[C]$ are of type~1.
            Indeed, if $\overrightarrow{g_ig_j}$ is an edge of type~2, then $g_j \in B$, and if $\overrightarrow{g_ig_j}$ is an edge of type~3, then $g_i \in A$.

            It follows that $\overrightarrow{g_i g_j}$ is an edge in $H[C]$ if and only if $e_i$ crosses~$e'_j$.
            In order to find a $g_j \in C$ that has no incoming edge, we shall resort to basic geometry:             
            Consider the \representation of $T$ and $T'$ with horizontal spine, but restricted only to the edges corresponding to gaps in $C$.
            For each $g_i \in C$ with pair $(e_i,e'_i)$, let $Z_i \sse \mathbb{R}^2$ be the union of the two open semicircles for $e_i$ and $e'_i$ (each shortened a bit at both endpoints).
            Due to the shortening, the elements of $\{Z_i \mid g_i \in C\}$ are pairwise disjoint; see \cref{fig:ABC-acyclicC}.

            \begin{figure}[ht]
                \centering
                \begin{subfigure}{.3\textwidth}
                    \centering
                    \includegraphics[page=5]{Intro}
                    \caption{}
                    \label{fig:ABC-acyclicA}
	        \end{subfigure}\hfil
                \begin{subfigure}{.6\textwidth}
                    \centering
                    \includegraphics{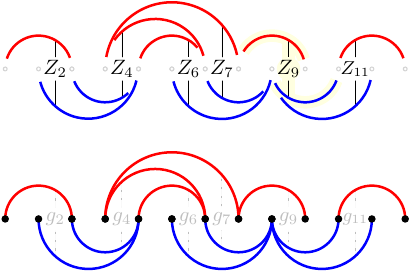}
                    \caption{}
                    \label{fig:ABC-acyclicC}
                \end{subfigure}
                \caption{
                    Illustration for the proof of \cref{lem:ABC-acyclic}.
                    (a) Case $Y=A$: Each of the red dashed edges is not there by the choice of $e_j$.
                    (b) Case $Y=C$: Gaps $g_2,g_4,g_6,g_7,g_9,g_{11}$ in $C$ and the corresponding $\{Z_i \mid g_i \in C\}$. Only $Z_9$ is not covered by any other $Z_i$.
                }
                \label{fig:C-is-acyclic}
            \end{figure}

            We say that a set $Y \sse \mathbb{R}^2$ is \emph{above} a set $X \sse \mathbb{R}^2$ if at least one $y \in Y$ lies vertically above at least one $x \in X$.
            Crucially, if $\overrightarrow{g_i g_j}$ is an edge in $H[C]$, then $Z_i$ is above $Z_j$.
            We claim that for some $Z_j$ no $Z_i$ is above $Z_j$, and hence $g_j \in C$ has no incoming edges in $H[C]$.
            To find $Z_j$, consider the rightmost point $r(Z_i)$ of each $Z_i$.
            (Using slight perturbations, we may assume that the $r(Z_i)$ have pairwise distinct $x$-coordinates.)
            There is at least one $r(Z_j)$ with no $Z_i$ above it; near $p_n$ at the latest.
            Among all rightmost points with no $Z_i$ above it, we pick the leftmost; say it is $r(Z_j)$ belonging to $g_j \in C$.
            Suppose (for the sake of a contradiction) some $Z_i$ is above $Z_j$.
            Then consider a rightmost point $q \in Z_j$ with some $Z_i$ above $q$.
            By the choice of $Z_j$, we have $q \neq r(Z_j)$.
            But then the rightmost point $r(Z_i)$ of $Z_i$ must be above $q$, and, by the choice of $q$, $r(Z_i)$ has no $Z_k$ above it --- a contradiction to the choice of $Z_j$, since $r(Z_i)$ lies further left than $r(Z_j)$.
            Hence, no $Z_i$ is above $Z_j$, and thus $g_j$ has no incoming edge $\overrightarrow{g_i g_j}$ in $H[C]$.
    \end{description}

    In each case $Y \in \{A,B,C\}$ we have found a gap $g^* \in Y$ with no incoming edges in $H[Y]$.
    Removing $g^*$ and repeating the argument shows that $H[Y]$ is acyclic.
\end{proof}

Together, \cref{lem:ABC-acyclic} and \cref{thm:alpha}\eqref{item:alpha-upper-bound} immediately imply an upper bound on~$\diam(\mathcal{F}_n)$, which, up to a small additive constant, is the same as~\cref{thm:main-upper-bound}. The full statement of~\cref{thm:main-upper-bound} follows from~\cref{thm:careful-main-upper-bound}.

\begin{corollary}
    \label{cor:5-over-3-upper-bound}
    Let $T,T'$ be any pair of two non-crossing trees on a convex set of $n$ points.
    Then the flip distance $\dist(T,T')$ is at most $\nicefrac53\cdot(n-1)$.
    In other words, $\diam(\mathcal{F}_n)\leq \nicefrac53\cdot(n-1)$.
\end{corollary}
\begin{proof}
    If $H$ is empty, then \cref{thm:alpha}\eqref{item:alpha-upper-bound} guarantees a flip distance of at most $\nicefrac32\cdot(n-1)$.
    If $H$ is non-empty, then \cref{lem:ABC-acyclic} states that $\frac{\las(H)}{|V(H)|} \geq \nicefrac13$ and \cref{thm:alpha}\eqref{item:alpha-upper-bound} implies an upper bound of $\max\left\{\nicefrac32, 2-\nicefrac13\right\}\cdot(n-1) = \nicefrac53\cdot(n-1)$.
\end{proof}

\subsection{Lower Bound on the Flip Distance via \cref{thm:alpha}\eqref{item:alpha-lower-bound}}
\label{subsec_lower}

We present an example of two trees $T,T'$ where the largest acyclic subset in the corresponding conflict graph $H$ is comparatively small.
By \cref{thm:alpha}\eqref{item:alpha-lower-bound} this then improves the lower bound on the diameter of $\mathcal{F}_n$ from $1.5 \cdot n - \mathcal{O}(1)$ to $\nicefrac{14}{9} \cdot n - \mathcal{O}(1) = 1.\overline{5} \cdot n - \mathcal{O}(1)$.

\begin{lemma}
    \label{lem_example_14_9}
    There exist trees $T,T'$ on linearly ordered points $p_1,\ldots,p_{13}$ such that for their conflict graph $H$ we have $\las(H)\leq 4$ and $|V(H)|=9$. 
\end{lemma}
\begin{proof}
    Let $T,T'$ be the trees depicted in \cref{fig:exampleLBa}.
    Their conflict graph $H$ is depicted in \cref{fig:exampleLBb} and contains a cycle of length $9$ with all edges being bi-directed.
    Consequently, any acyclic subset may contain at most every other gap and thus $\las(H)\leq\lfloor\nicefrac{9}{2}\rfloor=4$.
\end{proof}

\begin{figure}[htb]
    \centering
    \begin{subfigure}{.5\textwidth}
        \centering
	\includegraphics[page=1]{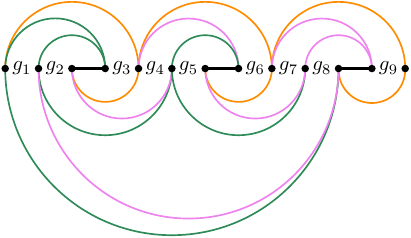}
	\caption{}
        \label{fig:exampleLBa}
    \end{subfigure}\hfil
    \begin{subfigure}{.4\textwidth}
	\centering
	\includegraphics[page=2]{flipping-example}
	\caption{}
        \label{fig:exampleLBb}
    \end{subfigure}
    \caption{
        Illustration for the proof of \cref{lem_example_14_9}.
        (a) A \representation of two trees $T,T'$ and (b) their conflict graph $H$.
        The coloring of edge pairs in (a) and gaps in (b) is according to the partition $A$, $B$, $C$ (for above, below, crossing as defined in \cref{subsec_upper}) in orange, green, and pink, respectively.
    }
    \label{fig:exampleLB}
\end{figure}

Together, \cref{lem_example_14_9} and \cref{thm:alpha}\eqref{item:alpha-lower-bound} imply \cref{thm:main-lower-bound}:

\mainlowerbound*

\section{Proof of \cref{thm:alpha}}
\label{sec_proof_alpha}

We now prove \cref{thm:alpha}, which relates the size of acyclic subsets of conflict graphs with upper and lower bounds for $\diam(\mathcal{F}_n)$, and is the key ingredient to \cref{thm:careful-main-upper-bound,thm:main-lower-bound,thm:main-upper-bound}.
We prove \cref{thm:alpha}\eqref{item:alpha-upper-bound} in \cref{subsec_upper_alpha} and \cref{thm:alpha}\eqref{item:alpha-lower-bound} in \cref{subsec_lower_alpha}.

\subsection{Upper Bound}
\label{subsec_upper_alpha}

Recall that we want a flip sequence that transforms $T$ into $T'$.
With each flip we remove an edge and replace it by a new edge.
This way, we can trace each edge from its initial to its final position.
In particular, every flip sequence naturally pairs the edges of $T$ with those of~$T'$. 
Our approach for the upper bound is to let $\pairs$ be this pairing, i.e., to convert each edge $e$ of $T$ into the edge $e'$ of $T'$ with $(e,e') \in \pairs$.
For each pair $(e,e') \in \pairs$ we will do at most two flips.
More precisely, in our flip sequence, every gap $g_i$, $i \in [n-1]$ and the corresponding pair $(e_i,e'_i) \in \pairs$ shall have exactly one of the following properties.
\begin{description}
    \item[$0$-flip:] 
        $e_i = e'_i$ and the edge $e_i$ is never replaced, keeping $e_i=e'_i$ in every intermediate tree.
    \item[$1$-flip:]
        $e_i \neq e'_i$ and the edge $e_i$ is replaced by $e'_i$ in a single, direct flip.
    \item[$2$-flip:] 
        $e_i \neq e'_i$, the edge $e_i$ is replaced by the boundary edge $p_ip_{i+1}$ in one flip, and $p_ip_{i+1}$ is replaced by $e'_i$ in a later flip.    
\end{description}
Clearly, the total number of flips in our flip sequence is then the number of $1$-flips plus two times the number of $2$-flips.
Our goal is to have as few $2$-flips as possible.

Recall the partition of $\pairs$ into $\pairs_=,\pairs_N,\pairs_R$.
As mentioned before, we shall spend no flips on pairs in $\pairs_=$ and in total at most $\frac32|\pairs_R| + |\pairs_=|$ flips on pairs in $\pairs_R$.
For the pairs in $\pairs_N$, we shall do a $1$-flip for those corresponding to an acyclic subset of the conflict graph $H$, and spend a $2$-flip for the remaining pairs in $\pairs_N$.
But first, let us present sufficient conditions for the validity of these flip.

Recall that every non-crossing edge-set $E$ on linearly ordered points $p_1,\ldots,p_n$ that covers all gaps has a corresponding map $\rho_E \colon \{g_1,\ldots,g_{n-1}\} \to E$, and that by \cref{lem_gap_edge_bij}, the set $E$ forms a tree if and only if $\rho_E$ is a bijection.

\begin{lemma}
    \label{lem:validFlip}
    Let $T$ be a non-crossing tree on linearly ordered points $p_1,\ldots,p_n$, and let $e_k = \rho_{T}(g_k)$ for $k = 1,\ldots,n-1$.
    
    Fix an edge $e_j \in T$ and consider an edge $e' = p_xp_y$ with $e' \notin T$, such that $e'$ covers $g_j$, and $e'$ does not cross any edge in $T - e_j$, and there is no $e_i\in T - e_j$ such that
    \begin{alphaenumerate}
        \item $e'$ covers $e_i$, and $e_i$ covers $g_j$, or \label{item:require-a}
        \item $e_i$ covers $e'$, and $e'$ covers $g_i$.\label{item:require-b}
    \end{alphaenumerate}
    Then, for $T' := (T - e_j) + e'$, each of the following holds.
    \begin{romanenumerate}
        \item $T'$ is a non-crossing tree, i.e., $\rho_{T'}$ is a bijection.\label{item:flip-non-crossing}
	
        \item Each edge $e \in T \cap T' = T - e_j$ corresponds to the same gap in $T$ and $T'$ ($\rho_{T}^{-1}(e) = \rho_{T'}^{-1}(e)$); the edge $e'$ corresponds to $g_j$ in $T'$ ($e' = \rho_{T'}(g_j)$).\label{item:flip-new-gaps}
        
        \item Each edge $e \in T \cap T' = T - e_j$ is \eshort{} (respectively \enormal{}, \eweird{}) in $T$ if and only if $e$ is \eshort{} (respectively \enormal{}, \eweird{}) in $T'$.\label{item:flip-same-types}
    \end{romanenumerate}
\end{lemma}
\begin{proof}
    We first show \eqref{item:flip-non-crossing}.
    Note that $T'$ is non-crossing as $e'$ crosses no edge in $T - e_j$ by assumption.
    To show that $T'$ is a tree, by \cref{lem_gap_edge_bij}, it suffices to show that $\rho_{T'} \colon \{g_1,\ldots,g_{n-1}\} \to T'$ is a bijection.
    In fact, gap $g_j$ is covered by $e'$ and each $g_i \neq g_j$ is still covered by $e_i \in T-e_j$.
    As $|T'| = n-1$, it suffices to show that $\rho_{T'}$ is injective.
    By \eqref{item:require-a}, $\rho_{T'}(g_j) = e'$, and by \eqref{item:require-b}, $\rho_{T'}(g_i) = e_i$ for all $i\neq j$.
    Thus, $\rho_{T'}$ is a bijection and $T'$ a tree.

    In fact, we already know $\rho_{T'}$ explicitly, and can also conclude \eqref{item:flip-new-gaps}.
    Finally, \eqref{item:flip-same-types} follows from~\eqref{item:flip-new-gaps}, since the type of an edge $e$ depends only on $e$ and its associated gap.
\end{proof}

We use \cref{lem:validFlip} in particular for two special cases, namely, for $2$-flips when $e'$ is a boundary edge and for $1$-flips when the gap of $e'$ has no incoming edges in a subgraph $H[Y]$ of $H$.
For convenience, we show that the preconditions are fulfilled in these two cases.

\begin{lemma}\label{lem:flip-to-boundary}
    The preconditions of \cref{lem:validFlip} are fulfilled if we choose $e'$ as the boundary edge $p_jp_{j+1}$.
\end{lemma}
\begin{proof}
    Clearly, a boundary edge does not cross any edge of $T$.
    Moreover, $e'$ is a \eshort{}edge that covers no edge of $T$ and covers only one gap, namely $g_j$, proving that no edge $e_i \in T - e_j$ has property \eqref{item:require-a} or \eqref{item:require-b}.
\end{proof}

\begin{lemma}\label{lem:validFlipAc}
    The preconditions of \cref{lem:validFlip} are fulfilled if $\pairs_R=\emptyset$, $g_j$ has no incoming edge in $H$ and $e'$ is chosen such that $(e_j,e')\in \pairs_N$, i.e., we flip $e_j$ for $e' = e'_j$.
\end{lemma}
\begin{proof}
    Note that no \eshort{} edge of $T$ can violate a precondition of \cref{lem:validFlip}, since any such edge can cover neither $e'$ nor $g_j$.
    As $\pairs_R=\emptyset$ by assumption, the only edges of $T$ that remain to be checked are edges $e_i$ that belong to a pair $(e_i,e'_i)\in \pairs_N$.
    If $e'$ crosses such an edge $e_i$ of $T - e_j$, then there is the incoming edge $\overrightarrow{g_ig_j}$ of type~1 at $g_j$ in $H$.
    Secondly, if an edge $e_i$ of $T - e_j$ has property \eqref{item:require-a}, respectively \eqref{item:require-b}, then there is the incoming edge $\overrightarrow{g_ig_j}$ of type~2, respectively type~3, at $g_j$ in $H$.
\end{proof}

Recall that for a gap $g_i$ and non-crossing tree $T$, the edge $e_i = \rho_T(g_i)$ is \eweird{} if $e_i$ is neither incident to $p_i$ nor $p_{i+1}$.
We next bound the number $|\TW|$ of \eweird{} edges of $T$ in terms of the number $|\TS|$ of \eshort{} edges in $T$. 

\begin{lemma} 
   \label{lem_tiny_vs_weird}
    Let $T$ be a non-crossing tree on linearly ordered points $p_1,\ldots,p_n$. 
    Let $k\geq 1$ be the number of uncovered edges of $T$.
	Then $|\TS| \geq k$ and $|\TW| \leq |\TS|-k$.
\end{lemma}
\begin{proof} 
    Consider the cover relation of the edges of $T$ (with respect to the given linear order).
    Let us write $e \preceq f$ whenever $e$ is covered by $f$.
    Trivially, every edge $e \in T$ covers itself, i.e., $e \preceq e$.
    Since $e \preceq e'$ and $e' \preceq e''$ implies $e \preceq e''$, we have that $(T, \preceq)$ forms a partial order. 
    Moreover, since $T$ is non-crossing, it holds that for every $e\in T$, its upset $\{e' \in T \mid e \preceq e'\}$ is totally ordered, implying that the Hasse diagram of $(T,\preceq)$ is a rooted forest~$R$.
    The roots of~$R$ are the $k$ uncovered edges of $T$. 
    The nodes without children of~$R$ are the \eshort{} edges of $T$. 
    For an example, consider \cref{fig:rootedForest}.
    \begin{figure}[htb]
	\centering
	\includegraphics[page=26]{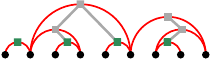}
        \caption{
            Illustration for the proof of \cref{{lem_tiny_vs_weird}}.
            Example of a tree $T$ and the corresponding rooted forest $R$ consisting of three trees. 
            The root of the middle tree of $R$ is a \eweird{} edge of $T$.
        }
        \label{fig:rootedForest}
    \end{figure}
    For convenience, we refer to any node without children as a leaf, even if it also the root of a tree.
    {Clearly, every rooted tree has at least one leaf}. 
    Consequently, $|\TS| \geq k$ follows from the fact that  $R$ consists of $k$ trees.
	
    To prove the second statement of the lemma, namely $|\TW| \leq |\TS|-k$, we show that any \eweird{} edge of~$T$ has at least two children in $R$. 
    By the well-known fact that the number of leaves in a rooted tree is at least one plus the number of nodes with at least two children, and our earlier observation that the {leaves} of~$R$ are the \eshort{} edges of $T$, $|\TW| \leq |\TS|-k$ follows.
 
    Let $e = p_ip_j$ be a \eweird{} edge with gap $g_k$ between $p_k$ and $p_{k+1}$. 
    Then $i<k$ and $k+1<j$.
    Let $e_i = \rho_T(g_i)$ be the edge with gap $g_i$, and let $e_{j-1} = \rho_T(g_{j-1})$ be the edge with gap $g_{j-1}$.
    The edges $e$, $e_i$, and $e_{j-1}$ are all different, since they have pairwise different gaps.
    Since $i<k$ and $k+1<j$, we have $e_i \preceq e$ and $e_{j-1} \preceq e$.
    Furthermore, any edge $f\neq e$ that covers both $e_i$ and $e_{j-1}$ also covers the points $p_k$ and $p_{k+1}$.
    However, since $g_k$ is the gap of $e$, it follows that $f$ also covers $e$.
    Hence, $e_i$ and $e_{j-1}$ are children of $e$ in $R$.
\end{proof}

Recall that we plan to flip some edges $e_i$ of $T$ to the boundary edge $p_ip_{i+1}$ corresponding to the gap $g_i$ of $e_i$.
In general, for a subset $I$ of gaps, let us denote by $T_I$ the graph obtained from $T$ by replacing, for each gap $g_i\in I$, the edge $e_i$ of $T$ by the corresponding boundary edge $p_ip_{i+1}$.
If $e_i$ is already \eshort{}, then $e_i = p_ip_{i+1}$; i.e., this replacement does not change anything and hence no flip os required.
Otherwise, $e_i \to p_ip_{i+1}$ always constitutes a valid flip by~\cref{lem:flip-to-boundary}.
In particular, \cref{lem:validFlip,lem:flip-to-boundary} assert that $T_I$ is a non-crossing tree and that there is a valid flip sequence $T \to \cdots \to T_I$.
The length of this flip sequence is the number of gaps in $I$ \emph{not} corresponding to \eshort{} edges of $T$. 

The following proposition states the number of flips that are required for handling the edges corresponding to pairs $\pairs_R \cup  \pairs_= \cup X$, where $X$ corresponds to the subset of edge pairs in $\pairs_N$ for which we make a $2$-flip in our final flip sequence.
Note that the proposition gives two partial flip sequences, since we will make all $1$-flips for pairs in  $\pairs_N$ between the first and the second flip of any $2$-flip of a pair in $\pairs$. 
This way, we will be able to use~\cref{lem:validFlipAc} for the pairs in $\pairs_N \setminus X$.
Recall that $I_=,I_R,I_N$ are the subsets of gaps corresponding to $\pairs_=,\pairs_R,\pairs_N$, respectively. 
We will especially consider the gaps in $I_R$ that correspond do a \eshort{}edge in $T$ or $T'$, since for those, we will require only 1 flip in our final flip sequence.

\begin{proposition}
    \label{prop:flip-to-boundary}
    Let $T,T'$ be two non-crossing trees on linearly ordered points $p_1,\ldots,p_n$, and $X \sse I_N$ be any (possibly empty) subset of the gaps corresponding to the \enormal{}-\enormal{} pairs.
    Then there are flip sequences $T \to \cdots \to T_{I_R \cup X}$ and $T'_{I_R \cup X} \to \cdots \to T'$ with in total at most $\frac32 |I_R| + |I_=| + 2|X| - 1$ flips.
\end{proposition}
\begin{proof}
    Consider any $g_i \in I_R \cup X$ and the corresponding pair $(e_i,e'_i)$.
    If $g_i \in I_R$, at most one of $e_i,e'_i$ is \eshort{}.
    If $g_i \in X \subseteq I_N$, none of $e_i,e'_i$ is \eshort{}.
    Let $\IRS = \{g_i \in I_R \mid \text{$e_i$ is \eshort{}}\}$ and $\IRSp = \{g_i \in I_R \mid \text{$e'_i$ is \eshort{}}\}$.
    Then $\IRS \cap \IRSp = \emptyset$, since \eshort{}-\eshort{} pairs are in $I_=$ only.
    For the first flip sequence $T \to \cdots \to T_{I_R \cup X}$ we do (in any order) for every $g_i \in (I_R \cup X) - \IRS$ a flip that replaces $e_i$ by $p_ip_{i+1}$.
    This is a valid flip sequence by \cref{lem:flip-to-boundary}, and clearly transforms $T$ into $T_{I_R \cup X}$.
    It uses $|(I_R \cup X) - S| = |I_R - \IRS| + |X| = |I_R| - |\IRS| + |X|$ flips.
    Similarly, there is a valid flip sequence $T' \to \cdots \to T'_{I_R \cup X}$ that uses $|I_R| - |\IRSp| + |X|$ flips.
    Its reverse is the desired flip sequence $T'_{I_R \cup X} \to \cdots \to T'$.  

    In total, both flip sequences have $2|I_R| - (|\IRS| + |\IRSp|) + 2|X|$ flips.
    It remains to prove that $|\IRS| + |\IRSp| \geq \frac12|I_R| - |I_=| + 1$.
    To this end, note that every gap in $I_R - (\IRS \cup \IRSp)$ involves at least one \eweird{} edge.
    Recall that $\TW$ and $\TS$ denote the set of all \eweird{} and all \eshort{} edges in $T$, respectively.
    By \cref{lem_tiny_vs_weird}, we have $|\TW| \leq |\TS| - 1$ and $|\TWp| \leq |\TSp| - 1$.
    Moreover, $|\TS| + |\TSp| \leq |\IRS| + |\IRSp| + 2|I_=|$.
     Together we conclude 
    \[
        |I_R| \leq |\IRS| + |\IRSp| + |\TW| + |\TWp| \leq |\IRS| + |\IRSp| + |\TS| + |\TSp| - 2 \leq 2(|\IRS| + |\IRSp| + |I_=|) - 2,
    \]
    which gives the desired $|\IRS| + |\IRSp| \geq \frac12|I_R| - |I_=| + 1$.
\end{proof}

\cref{prop:flip-to-boundary} works for any subset $X \sse I_N$ of the \enormal{}-\enormal{} gaps.
We content ourselves with spending a $2$-flip on each gap in $X$ (reflected by the $2|X|$ term in the bound of \cref{prop:flip-to-boundary}), but aim to do a direct $1$-flip on each gap in $Y = I_N - X$.
For larger $|Y|$ we obtain an overall shorter flip sequence.
So we want a large set of \enormal{}-\enormal{} pairs that can all be done as $1$-flips.
These flips shall form a valid flip sequence $T_{I_R \cup X} \to \cdots \to T'_{I_R \cup X}$, connecting the two sequences obtained by \cref{prop:flip-to-boundary}.
As all the edges for gaps in $I_R \cup X$ are flipped to boundary edges in $T_{I_R \cup X}$ and $T'_{I_R \cup X}$, we can ``safely ignore'' all pairs corresponding to gaps in $I_= \cup I_R \cup X$ and focus on the \enormal{}-\enormal{} pairs corresponding to~$Y$.

\begin{proposition}
    \label{prop:flip-NN}
    Let $T,T'$ be two non-crossing trees on linearly ordered points $p_1,\ldots,p_n$, and $Y \sse I_N$ be an acyclic subset in $H = H(T,T')$, and $X=I_N-Y$.
    Then there is a flip sequence $T_{I_R \cup X}\to \dots \to T'_{I_R \cup X} $ with $|Y|$ flips.
\end{proposition}
\begin{proof}
    Let $T_1:=T_{I_R \cup X}$ and $T_2:=T'_{I_R \cup X}$.
    \cref{lem:validFlip} guarantees that \enormal{}-\enormal{} pairs of $T,T'$ corresponding to gaps in $Y$ are still \enormal{}-\enormal{} pairs of $T_1,T_2$.
    Moreover, we have $\pairs_R=\emptyset$ for $T_1,T_2$.
    In particular, the conflict graph of $T_1,T_2$ is $H[Y]$.
    Because $H[Y]$ is acyclic, there exists a topological ordering $\prec$ of $H[Y]$.
    The first gap $g$ in $\prec$ has no incoming edges in $H[Y]$, and \cref{lem:validFlipAc,lem:validFlip} ensure that the direct flip of the corresponding pair $(e,e')$ of $g$ is valid and maintains all gap-assignments and types of edges.
    We repeat with direct flips for all pairs corresponding to $Y$ in the order given by $\prec$, until we reach $T_2$.
    As we spent one flip per pair, the resulting flip sequence has length $|Y|$.
\end{proof}

\subsubsection{Putting Things Together -- Proof of \cref{thm:alpha}\eqref{item:alpha-upper-bound}}

Now, we show how to obtain a short flip sequence from a large acyclic subset.

\begin{theorem*}[corresponding to \cref{thm:alpha}\eqref{item:alpha-upper-bound}]{\ \\}
    Let $T,T'$ be two non-crossing trees on linearly ordered points $p_1,\ldots,p_n$ with conflict graph $H = (V(H),E(H))$.
    Then the flip distance $\dist(T,T')$ is at most \mbox{$\max\left\{\frac32,2-\frac{\las(H)}{|V(H)|}\right\} (n-1)$} if $H$ is non-empty, and at most $\frac32(n-1)$ if $H$ is empty.
    \label{thm:large-acyclic-small-distance}
\end{theorem*}
\begin{proof}
    First assume that $H$ is non-empty.
    Let $Y \sse I_N = V(H)$ be an acyclic subset of $H$ with $|Y| = \las(H)$.
    Let $\prec$ be a topological ordering of $H[Y]$.
    Denoting $X = I_N - Y$, our flip sequence $F$ from  $T $ to $T'$ is composed of three parts:
    \begin{description}
        \item[$F_1$:]
            $T \to \cdots \to T_{I_R \cup X}$ replaces (in any order) each non-\eshort{} edge $e_i \in T$ with $g_i \in I_R \cup X$ by the boundary edge $p_ip_{i+1}$.
        
        \item[$F_2$:]
            $T_{I_R \cup X} \to \cdots \to T'_{I_R \cup X}$ replaces in order according to $\prec$ each edge $e_i \in T$ with $g_i \in Y$ by the edge $e'_i \in T'$.
        
        \item[$F_3$:]
            $T'_{I_R \cup X} \to \cdots \to T'$ replaces (in any order) each boundary edge $p_ip_{i+1}$ with $g_i \in I_R \cup X$ and $p_ip_{i+1} \notin T'$ by the non-\eshort{} edge $e'_i \in T'$.
    \end{description}
    
    By \cref{prop:flip-to-boundary}, the sequences $F_1$ and $F_3$ are valid and have a total length of $|F_1|+|F_3|=\frac32 |I_R| + |I_=| + 2|X| - 1$.
    \cref{prop:flip-NN} ensures that $F_2$ is valid and has length $|Y|=|I_N| - |X|$.
    With $|Y| = \las(H)$ and $I_N=V(H)$ we conclude that
    \begin{align*}
        \dist(T,T') &\leq |F_1|+|F_2|+|F_3|\leq \frac32|I_R| + |I_=| + |I_N| + |X| = \frac32|I_R| + |I_=| + 2|I_N| - |Y| \\
        &\leq \frac32(|I_R|+|I_=|) + \left(2-\frac{|Y|}{|I_N|}\right)|I_N| = \frac32(|I_R|+|I_=|) + \left(2-\frac{\las(H)}{|V(H)|}\right)|I_N|\\
        &\leq \max\left\{\frac32,2-\frac{\las(H)}{|V(H)|}\right\} (|I_R|+|I_=|+|I_N|) \leq \max\left\{\frac32,2-\frac{\las(H)}{|V(H)|}\right\} (n-1).
    \end{align*}
    If $H$ is empty, then $\pairs_N = \emptyset$, and hence, $|I_N| = |X| = 0$.
    Then the above with only $F_1$ and $F_3$ gives $\dist(T,T') \leq \frac32|I_R| + |I_=| \leq \frac32(n-1)$.
\end{proof}

\subsection{Lower Bound}
\label{subsec_lower_alpha}

In this section, we show that a single example of a pair $T,T'$ of trees gives rise to a lower bound for $\diam(\mathcal{F}_n)$ for all $n$ through properties of the conflict graph $H$ of $T,T'$.
To be precise, we show the following statement, which corresponds to \cref{thm:alpha}\cref{item:alpha-lower-bound}.

\begin{restatable}{theorem}{DELTA}  
    \label{thm_Delta}
    Let $T,T'$ be non-crossing trees on linearly ordered points $p_1,\ldots,p_n$ with corresponding conflict graph $H = H(T,T')$.
    If $V(H)$ is non-empty, then there is a constant $\constantC$ depending only on $T$ and~$T'$ such that for all $\largeN\geq 1$ we have $\diam(\mathcal{F}_{\largeN}) \geq (2-\frac{\las(H)}{|V(H)|})\largeN-\constantC$.
\end{restatable}

To this end, let us consider a pair of trees $T$ and $T'$ on $n$ points with conflict graph $H$. 
We will construct a sequence of tree pairs $(\blowup{T},\blowup{T'})_{k\in \mathbb N}$ on $n_k$ points each such that
\[
    \dist(\blowup{T},\blowup{T'}) \geq \left(2-\frac{\las(H)}{|V(H)|}\right)n_k-c. 
\]
We now explain how to construct these trees.
We consider the edge pairs $\pairs$ of $T$ and $T'$ as before.
Recall that $\pairs_N$ is the set of \enormal{}-\enormal{} pairs $(e,e')$ of $T$ and $T'$ with $e\neq e'$, and that this set is assumed to be nonempty.
Let $\ennr$ and $\ennb$ be the sets of edges of $T$ and $T'$, respectively, appearing in pairs in $\pairs_N$. 
For $k\geq 1$, we define the $k$-\emph{blowups} \blowup{T} and \blowup{T'} of $T$ and $T'$ by doing the following for each $(e,e')\in \pairs_N$ (for an illustration consider \cref{fig:blowup}).
\begin{itemize}
    \item Insert a set $V(e)$ of $k$ new points in the gap $g$ associated to $(e,e')$.
    \item In $\blowup{T}$, add an edge from each $v\in V(e)$ to the endpoint of $e$ that is not adjacent to $g$, and similarly for $T'$ and $e'$.
    \item Let $\Lambda(e)$ denote the set of the $k$ edges added to $\blowup{T}$, and let $\Lambda(e')$ denote the set of the $k$ edges added to $\blowup{T'}$.
\end{itemize}
This way, for each edge $e$ appearing in a pair in $\pairs_N$ we add next to $e$ a fan $\Lambda(e)$ of $k$ edges ending at leaves.
By construction, the blowups have $n_k := n+k|\pairs_N|$ points and $n_k-1$ edges.

\begin{figure}[htb]
    \centering
    \begin{subfigure}{.5\textwidth}
        \centering
        \includegraphics[page=3]{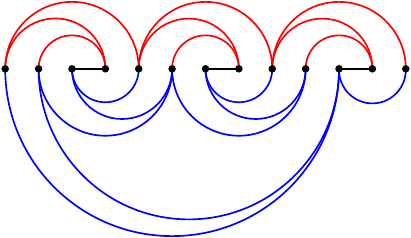}
        \caption{}
    \end{subfigure}\hfil
    \begin{subfigure}{.5\textwidth}
        \centering
        \includegraphics[page=4]{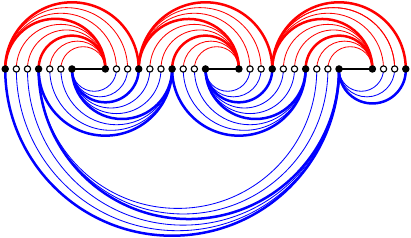}
        \caption{}
    \end{subfigure}
    \caption{A pair of trees $(T,T')$ and their 2-blowups $(\blowup[2]{T},\blowup[2]{T'})$.}
    \label{fig:blowup}
\end{figure}

Here is the crucial connection between $k$-blowups and the conflict graph $H$ of $T,T'$.

\begin{observation}
    \label{obs:crossing-in-blowups}
    If $\overrightarrow{g_ig_j}$ is a directed edge in the conflict graph $H$ of $T,T'$, then in the $k$-blowups $\blowup{T},\blowup{T'}$ every edge in $\Lambda(e_i)$ crosses every edge in $\Lambda(e'_j)$.
\end{observation}
\begin{proof}
    If $\overrightarrow{g_ig_j}$ is of type~1, then already $e_i$ and $e'_j$ cross and hence so do their fans.
    For $\overrightarrow{g_ig_j}$ of type~2 and type~3, it follows from the fact that the points $V(e_i)$ and $V(e'_j)$ lie in the gaps $g_i$ and $g_j$ of the original point set, respectively; see~\cref{fig:crossing-in-blowup} for illustrations.
\end{proof}

\begin{figure}[ht]
    \centering
    \begin{subfigure}{.5\textwidth}
	\centering
	\includegraphics[page=1]{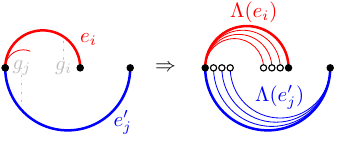}
	\caption{$\overrightarrow{g_i g_j}$ of type~2}
        \label{fig:crossing-in-blowup-type2}
    \end{subfigure}\hfil
    \begin{subfigure}{.5\textwidth}
        \centering
	\includegraphics[page=2]{crossing-in-blowup}
        \caption{$\overrightarrow{g_i g_j}$ of type~3}
        \label{fig:crossing-in-blowup-type3}
    \end{subfigure}
    \caption{A directed edge $\overrightarrow{g_ig_j}$ of type~2 (a) or type~3 (b), makes $\Lambda(e_i)$ crossing $\Lambda(e'_j)$.}
    \label{fig:crossing-in-blowup}
\end{figure}

Now, we consider any flip sequence~$F$ from $\blowup{T}$ to $\blowup{T'}$ and denote the intermediate trees by $\blowup{T} = \blowupSeq[0], \blowupSeq[1], \dots, \blowupSeq[\ell] = \blowup{T'}$.
For each $(e,e')\in \pairs_N$, let $\gone(e)$ be the smallest index $a$ such that $\blowupSeq[a]$ contains no edge in $\Lambda(e)$.
Since \blowup{T'} contains no edge in $\Lambda(e)$, $\gone(e)$ is well-defined.
Evidently, in $F$ there is a flip $\blowupSeq[\gone(e)-1] \to \blowupSeq[\gone(e)]$ that replaces the last remaining edge in $\Lambda(e)$ by an edge not in $\Lambda(e)$.
We say that the pair $(e,e')$ is \emph{direct} if there is an $a\leq \gone(e)$ such that $\blowupSeq[a]$ contains an edge in $\Lambda(e')$, and \emph{indirect} otherwise.
These terms are motivated by the fact that in the former case, $F$ might contain a direct flip from an edge in $\Lambda(e)$ to an edge in in $\Lambda(e')$, while in the latter case this is impossible. 

\begin{lemma}
    \label{lem_direct_pairs}
    Let $(e_i,e'_i)\neq (e_j,e'_j)$ be direct pairs in $\pairs_N$.
    If the conflict graph $H$ of $T,T'$ contains the directed edge $\overrightarrow{g_ig_j}$, then $\gone(e_i) < \gone(e_j)$.
\end{lemma}
\begin{proof} 
    We show that if $\gone(e_i)\geq \gone(e_j)$, then there is no edge from $g_i$ to $g_j$ in~$H$.

    Since $\gone(e_i)\geq \gone(e_j)$, $\blowupSeq[\gone(e_j)-1]$ contains an edge of $\Lambda(e_i)$.
    The edge that is flipped away from $\blowupSeq[\gone(e_j)-1]$ in the flip $\blowupSeq[\gone(e_j)-1] \to \blowupSeq[\gone(e_j)]$ is in $\Lambda(e_j)$, so since $\Lambda(e_j)\cap \Lambda(e_i)=\emptyset$ by construction, $\blowupSeq[\gone(e_j)]$ also contains an edge in $\Lambda(e_i)$.
    Thus, $\gone(e_i)> \gone(e_j)$.

    Now choose any $a\leq\gone(e_j)$ such that $\blowupSeq[a]$ contains an edge in $f' \in \Lambda(e'_j)$.
    Since $a \leq \gone(e_j) < \gone(e_i)$, $\blowupSeq[a]$ contains at least one edge of $f \in \Lambda(e_i)$.
    But $\blowupSeq[a]$ is non-crossing, so we have found edges $f\in \Lambda(e_i)$ and $f'\in \Lambda(e'_j)$ that do not cross.
    By \cref{obs:crossing-in-blowups} it follows that $\overrightarrow{g_ig_j}$ is not an edge in $H$.
\end{proof}

Let $\delta$ and $\bar \delta$ denote the number of direct and indirect pairs in $\pairs_N$ of the original pair $T,T'$ of trees induced by the flip sequence $F$ between the blowups $\blowup{T},\blowup{T'}$, respectively.
Clearly, $\delta+\bar \delta=|\pairs_N|$. 
\cref{{lem_direct_pairs}} implies the following crucial property.

\begin{corollary}
    \label{cor_direct_acyclic}
    The conflict graph $H$ has an acyclic subset of size $\delta$, i.e., $\las(H)\geq \delta$. 
\end{corollary}
\begin{proof}
    Let $(e_{i_1},e'_{i_1}), \ldots, (e_{i_\delta},e'_{i_\delta})$ be the direct pairs, sorted so that $\gone(e_{i_1})\leq \dots \leq \gone(e_{i_\delta})$.
    By \cref{lem_direct_pairs}, every edge of $H$ between the gaps of two direct pairs points forward in that ordering.
    Hence, the subgraph of $H$ corresponding to (the gaps of) direct pairs is acyclic.
\end{proof}

Now, we aim to show a lower bound on the flip sequence in terms of the largest acyclic subset in $H$.
Intuitively speaking, we show that for each indirect pair $(e,e')$, the process of removing the $k$ edges in $\Lambda(e)$ and adding the $k$ edges in $\Lambda(e')$ in the flip sequence $F$ must involve introducing almost $k$ ``intermediate'' edges that are neither in $\blowup{T}$ nor in $\blowup{T'}$, which then increases the length of $F$.
\cref{lem_green_graph} below shows that a single indirect pair gives rise to many intermediate edges, i.e., costs additional flips.
\cref{lem_thickened_flip_length} further below shows that costs for different indirect pairs add up.
That is, we cannot ``reuse'' intermediate edges to reduce the cost in any effective way.

\begin{lemma}
    \label{lem_green_graph}
    Let $(e,e')\in \pairs_N$ be an indirect pair.
    Then there is a subgraph $G(e,e')$ of the tree $\blowupSeq[\gone(e)]$ that contains $V(e)$, does not contain any edges in $\blowup{T}$ or $\blowup{T'}$, and has at most $2n-1$ connected components.
\end{lemma}
\begin{proof}
    Let $a=\gone(e)$.
    Because $(e,e')$ is indirect, $\blowupSeq[a] \cap \Lambda(e)$ and $\blowupSeq[a] \cap \Lambda(e')$ are empty.
    We first find a subgraph $E$ of $\blowupSeq[a]$ by doing the following for every edge $f \in \blowupSeq[a] \cap (\blowup{T} \cup \blowup{T'})$:
    \begin{itemize}
        \item Since $f \notin \Lambda(e)\cup \Lambda(e')$, all the points of $V(e)$ lie on the same side of $f$.
        \item Delete from $\blowupSeq[a]$ all the points (and their incident edges) that are on the other side (without $V(e)$) of $f$, keeping the endpoints of $f$.
    \end{itemize}
    Note that these deletions do not disconnect $\blowupSeq[a]$, so the remaining subgraph $E \sse \blowupSeq[a]$ is still connected.
    Further note that for every $f \in T \cup T'$ and its fan $\Lambda(f)$ in $\blowup{T}$ or $\blowup{T'}$, no two edges of $f \cup \Lambda(f)$ lie in $E$.
    Indeed, otherwise $V(e)$ lies on the same side of both these edges and one would be deleted when considering the other.
    Consequently, $E$ has at most $n-1$ edges of $\blowup{T}$ and at most $n-1$ edges of $\blowup{T'}$, i.e., $|E \cap (\blowup{T} \cup \blowup{T'})| \leq 2n-2$.
    Then $G(e,e') = E - (\blowup{T} \cup \blowup{T'})$ is the desired subgraph of $\blowupSeq[a]$.
\end{proof}

For the next lemma, recall that $F$ is an arbitrary but fixed flip sequence from $\blowup{T}$ to $\blowup{T'}$.

\begin{lemma}
    \label{lem_thickened_flip_length}
    The flip sequence $F$ from $\blowup{T}$ to $\blowup{T'}$ has length at least $(k-2n)(\bar \delta + |\pairs_N|)$.
\end{lemma}
\begin{proof}
    For each indirect pair $(e,e') \in \pairs_N$, let $G(e,e')$ be the corresponding subgraph of $\blowupSeq[\gone(e)]$ guaranteed by \cref{lem_green_graph}.
    Let $U$ be the union of all $G(e,e')$ over all indirect pairs $(e,e') \in \pairs_N$.
    By \cref{lem_green_graph}, $U$ has at most $(2n-1)\bar \delta$ connected components and at least $k\bar \delta$ points, because it contains the points $V(e)$ for all indirect pairs $(e,e')$, and $|V(e)|=k$.
    Thus, $U$ has at least $k\bar \delta - (2n-1)\bar \delta \geq (k-2n)\bar \delta$ edges, none of which is in $\blowup{T}$ or $\blowup{T'}$.
    
    Consequently, the flip sequence $F$ has at least $|U| \geq (k-2n)\bar \delta$ flips that introduce an edge of $U$, as well as $|\blowup{T'} \setminus \blowup{T}|$ additional flips that introduce an edge of $\blowup{T'} \setminus \blowup{T}$.
    For each $(e,e')\in \pairs_N$, there are $k$ edges in $\Lambda(e')$ that do not appear in $\blowup{T}$.
    Thus, $|\blowup{T'}\setminus \blowup{T}|\geq k|\pairs_N|$.
    Adding all together, $F$ has length at least $(k-2n)\bar \delta + k|\pairs_N| \geq (k-2n)(\bar \delta + |\pairs_N|)$.
\end{proof}

\begin{lemma}
    \label{lem_rho}
    For any fixed pair $T,T'$ of trees, there is a constant $\constantC$ depending only on $T$ and $T'$, such that for all $k\geq 1$ it holds that 
    \[
        \dist(\blowup{T},\blowup{T'}) \geq \left(2-\frac{\las(H)}{|V(H)|}\right)n_k - c.
    \]
\end{lemma}
\begin{proof}
    By \cref{lem_thickened_flip_length}, any flip sequence from $\blowup{T}$ to $\blowup{T'}$ with $\bar\delta$ indirect flips has length at least $(k-2n)(\bar \delta + |\pairs_N|)$.
    By construction, $\blowup{T}$ (and also $\blowup{T'}$) has $n_k = n+k|\pairs_N|$ points.
    Since $n$ is fixed and hence constant with respect to the total number $n_k$ of points, we get
    \begin{equation}
        \label{eq_vk}
        \dist(\blowup{T},\blowup{T'})\geq (k-2n)(\bar \delta + |\pairs_N|) \geq k(\bar \delta + |\pairs_N|) - c 
		\geq n_k \frac{\bar \delta + |\pairs_N|}{|\pairs_N|} - c. 
    \end{equation}
    The vertices of $H$ are in bijection with the pairs in $\pairs_N$, and by \cref{cor_direct_acyclic}, the number $\delta$ of direct pairs in $\pairs_N$ is at most $\las(H)$.
    Thus,
    \[
        \bar \delta = |\pairs_N| - \delta 
        \geq |\pairs_N|-\las(H) 
        = |V(H)|-\las(H).
    \]
    Plugging this with $|V(H)| = |\pairs_N|$ into \cref{eq_vk} completes the proof of the lemma:
    \[
        \dist(\blowup{T},\blowup{T'})\geq n_k \frac{2|V(H)|-\las(H)}{|V(H)|} - c 
		= n_k \left(2-\frac{\las(H)}{|V(H)|}\right) - c.\qedhere
    \]
\end{proof}

\cref{lem_rho} takes care of all $n$ of the form $n_k$ in \cref{thm_Delta}.
The following statement will help to deal with the intermediate values of $n$.

\begin{lemma}
    \label{lem_Fn_nondecreasing}
    The diameter of $\mathcal F_n$ is monotone in $n$, i.e., for any $n \ge 1$ it holds that $\diam(\mathcal{F}_n)\leq \diam(\mathcal{F}_{n+1})$.
\end{lemma}
\begin{proof}
    Let $p_1,\dots, p_{n+1}$ be points in convex position listed in counterclockwise order.
    Let $T$ and $T'$ be non-crossing trees on the vertices $p_1,\dots, p_n$ with $\dist(T,T') = \diam(\mathcal{F}_n)$.
    Let $T_+$ and $T'_+$ be trees on the vertices $p_1,\dots, p_{n+1}$, with $E(T_+) = E(T)\cup \{p_{n+1}p_1\}$ and $E(T'_+) = E(T')\cup \{p_{n+1}p_1\}$.
    Note that $T_+$ and $T'_+$ are non-crossing spanning trees.
    To prove the lemma, it suffices to show that $\dist(T,T')\leq \dist(T_+,T'_+)$, since this implies $\diam(\mathcal{F}_n)\leq \diam(\mathcal{F}_{n+1})$.

    It is known that any common boundary edge stays fixed in any shortest flip sequence between the trees~\cite[Corollary~18]{aichholzer2022reconfiguration}. 
    In particular, this is true for the edge $p_{n+1}p_1$ and any flip sequence between $T_+$ and $T'_+$ of length $\dist(T_+, T'_+)$.
    Take such a sequence and collapse $p_{n+1}$ and $p_1$ to $p_1$, removing the edge $p_{n+1}p_1$ at every step of the sequence.
    This gives a valid flip sequence between $T$ and $T'$, because this collapse does not change connectivity or the existence of cycles or crossings in the graphs.
    Moreover, this sequence has length $\dist(T_+, T'_+)$, showing $\dist(T,T')\leq \dist(T_+,T'_+)$, as desired.
\end{proof}

\begin{proof}[Proof of \cref{thm_Delta}]
    By \cref{lem_rho}, we have a family of pairs of trees $(\blowup{T},\blowup{T'})_{k \geq 1}$ showing the desired lower bound on $\diam(\mathcal{F}_{\largeN})$ for each ${\largeN}$ of the form ${\largeN} = n_k \coloneqq n+k|V(H)|$ for some $k\geq 1$.
    Since $n_{k+1} - n_k$ is constant, it suffices to show that $\diam(\mathcal{F}_{\largeN}) \geq \diam(\mathcal{F}_{n_k})$ for $n_k\leq {\largeN}< n_{k+1}$, which follows from \cref{lem_Fn_nondecreasing}.
\end{proof}

\subparagraph{Relation to the previous lower bound.}

One might wonder about a concrete instance where any flip sequence exceeds the previous lower bound.
The proof of \cref{thm:main-lower-bound} implies that for the trees $(T,T')$ depicted in~\cref{fig:exampleLBa} and $k$ large enough, $\dist(\blowup{T},\blowup{T'})$ is larger than the previous lower bound of $\lfloor \nicefrac{3}{2}\cdot n \rfloor -5$ from \cite[Theorem 3.5]{Hernando}, where $n=n_k$ is the number of vertices in $\blowup{T}$. 
But since the theorem involves a non-specified constant, it is not in itself enough to exhibit a concrete example of a pair of trees with a distance that is greater than the previous bound of $\lfloor \nicefrac{3}{2}\cdot n \rfloor -5$.
In the following, we show that for any $k>757$, $\dist(\blowup{T},\blowup{T'}) > \lfloor \nicefrac{3}{2}\cdot n_k \rfloor -5$.
That is, for any $k>757$, $(\blowup{T},\blowup{T'})$ is a pair of trees which beats the best previously known lower bound.

From \cref{lem_thickened_flip_length} and $\delta+\bar \delta = |\pairs_N|$ we know that for any pair $(T,T')$ of trees we have
$
    \dist(\blowup{T},\blowup{T'}) \geq (k-2n)(2|\pairs_N| - \delta).
$
Because $|V(H)| = |\pairs_N|$ and $\delta \le \las(H)$ by \cref{cor_direct_acyclic}, we have the following for $k-2n\ge 0$:
\[
    \dist(\blowup{T},\blowup{T'})\geq (k-2n)(2|V(H)| - \las(H)) 
\]
Since $\blowup{T}$ has $n_k = n+k|\pairs_N|=n+k|V(H)|$ points, the previous lower bound for $\dist(\blowup{T},\blowup{T'})$ is 
$
    \lfloor \nicefrac{3}{2}\cdot n_k \rfloor -5 \le \nicefrac{3}{2}\cdot n_k-5 = \nicefrac{3}{2}\cdot n + \nicefrac{3}{2}\cdot k|V(H)|-5.
$
Consequently, $\dist(\blowup{T},\blowup{T'})$ exceeds the previous lower bound whenever
\begin{equation}\label{eq:concrete_bound}
   (k-2n)(2|V(H)| - \las(H))> \nicefrac{3}{2}\cdot n + \nicefrac{3}{2}\cdot k|V(H)|-5.
\end{equation}
For our concrete trees $T,T'$ from~\cref{fig:exampleLBa}, we have $n=13$, $|V(H)|=9$, and $\las(H)=4$.
Plugging this into~\cref{eq:concrete_bound}, we get $(k-26)(18 - 4)> \nicefrac{3}{2}\cdot 13 + \nicefrac{3}{2}\cdot k\cdot 9-5$, which is equivalent to $k>757$.

We remark that we expect the true value for $k$ to be much smaller, since the the lower bound estimation from \cref{lem_thickened_flip_length,lem_rho} for the length of any flip sequence $F$ from $\blowup{T}$ to $\blowup{T'}$ is obviously far from tight; for $k < 2n$ that bound is is even negative. 
However, this non-optimal lower bound estimation only influences the constant additional term and not on the leading term of the new lower bound, which is one of the main results of this work.

\section{Keeping Common Edges and Improving the Upper Bound}
\label{sec_improved_upper}

In this section, we will further improve the upper bound of $\nicefrac 53\cdot (n-1)$ from \cref{cor:5-over-3-upper-bound} to also depend on the number of edges that the two trees share. 
Moreover, we will obtain a flip sequence that avoids flipping any edge that is already in both trees. 
We remark that in the flip sequences obtained for \cref{cor:5-over-3-upper-bound}, such flips might be performed. 
In particular, this will be the case whenever a common edge is assigned to different gaps for the different trees.

We distinguish two different types of edges in (a tree on) a convex point set $P$, namely, boundary edges and non-boundary edges, where we denote the latter as \emph{chords}.  

For a pair $T,T'$ of trees on $P$ let $b = b(T,T')$ denote the number of \emph{common edges} (edges in $T\cap T'$) that are also boundary edges.
Clearly $d+b \le n-1$.
\cref{thm:main-upper-bound} is implied by the following stronger statement.

\begin{theorem}
    \label{thm:careful-main-upper-bound}
    Let $T,T'$ be two non-crossing spanning trees on a set of $n\geq 3$ points in convex position.
    Let $d = |T-T'|$ and let $b$ be the number of common boundary edges of $T$ and~$T'$.
    Then $\dist(T,T')\leq \nicefrac 53 \cdot d+\nicefrac 23 \cdot b-\nicefrac 43$.
    Moreover, there exists a flip sequence from $T$ to $T'$ of at most that length in which no common edges are flipped.
\end{theorem}

The high level proof idea for \cref{thm:careful-main-upper-bound} is the following: 
We split the instance along common chords and thus obtain subinstances where all common edges are boundary edges, see \cref{fig:subinstances} for an illustration.
We handle these subinstances independently.
For each subinstance, we identitfy a ``good'' position to cut open the boundary of the point set to obtain a linear order of the points by the following two observations.

\begin{observation}\label{obs:minus_k}
    Let $T$ be a non-crossing spanning tree on a set $P$ of $n$ points in convex position. 
    Then for any boundary edge $p_1p_n$ of $P$ that is not a boundary edge of $T$, the tree $T$ with linear order $p_1,\ldots,p_n$ has at least two uncovered edges.
\end{observation}

The point order obtained by \cref{obs:minus_k} will avoid flipping common edges and will facilitate obtaining the upper bound of \cref{thm:careful-main-upper-bound} for each subinstance as well as in total. 

To use \cref{obs:minus_k}, we need to identify a boundary edge of $P$ that is not a boundary edge in any of the trees.
Note that, in particular after cutting along common edges, it is easy to see that one can perform a flip which introduces a boundary edge from $T'-T$ (or $T-T'$) and removes a chord, unless both trees consist of boundary edges only, see also Bousquet et al.~\cite[Claim~2]{bousquet2023noteJOURNAL}.
Hence we have the following observation.

\begin{observation}\label{lem:freeGapToCut}
    Consider two non-crossing spanning trees $T,T'$ on a convex point set $P$. If $T\cup T'$ contain all boundary edges of $P$, then there exists a flip sequence of length $|T\setminus T'|$.
\end{observation}

\begin{figure}[htb]
    \centering
    \includegraphics[page=25]{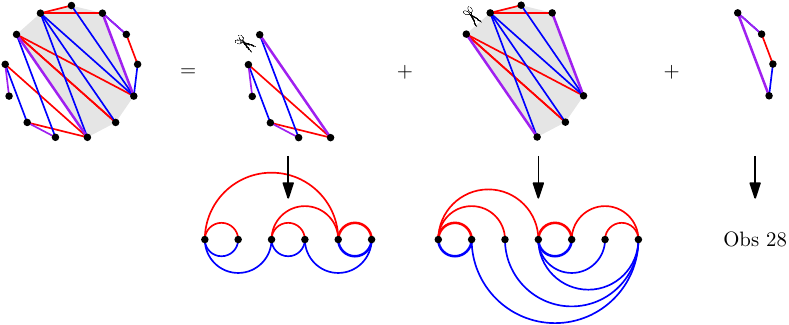}
    \caption{
        Illustration for the proof idea of \cref{thm:careful-main-upper-bound}. 
        The instance is split along common chords (in thick purple).
        For each subinstance, we either use \cref{lem:freeGapToCut} or choose a boundary edge of the point set (marked by a pair of scissors) that is not a boundary edge in any of the two trees in order to cut open the cycle and obtain a linear order on the points.
    }
    \label{fig:subinstances}
\end{figure}

We are now ready to prove \cref{thm:careful-main-upper-bound}.

\begin{proof}[Proof of \cref{thm:careful-main-upper-bound}]
    Let $B(T,T')$ be the set of common boundary edges of $T$ and~$T'$, with $|B(T,T')|=b$, let $C(T,T')$ be the set of common chords of $T$ and $T'$ with cardinality $|C(T,T')|=c$ and let $D(T,T')=T-T'$ (resp.~$D(T',T)=T'-T$) be the set of edges that are only in $T$ (resp.~$T'$) with cardinality $|D(T,T')|=d$. 
    Then $c=n-1-b-d$.

    Assume first that $c=0$, that is, that $T$ and $T'$ do not have any common chords. 

    By \cref{lem:freeGapToCut}, if $T \cup T'$ covers all boundary edges of $P$, then $\dist(T,T') = d$.
    Since every tree on $n\geq 3$ points has at least two boundary edges, $d=0$ then implies $b\ge2$ and $d=1$ implies $b \ge 1$. Hence $\dist(T,T') = d \le \nicefrac 53 \cdot d +\nicefrac 23 \cdot b-\nicefrac 43$.

    Otherwise, we may choose an arbitrary boundary edge $p_1p_n$ of $P$ that is neither a boundary edge of $T$ nor one of $T'$ to cut the cyclic order of the points of $P$ into a linear order $p_1,\ldots,p_n$. 
    With this linear order, each of ${T}$ and ${T}'$ has at least two uncovered edges by \cref{obs:minus_k}.

    Consider the set $\pairs$ of pairs of edges of $T$ and $T'$ that are induced by the gaps, its partition into $\pairs_=, \pairs_N, \pairs_R$, and the corresponding sets $I_=, I_N, I_R$ of gaps.
    Because the set $I_=$ consists of all gaps that correspond to \eshort{}-\eshort{} pairs, we have $|I_=|=|B(T,T')|=b$. 

    Recall that the set $I_R$ contains the gaps corresponding to all other pairs that contain at least one \eshort{} or \eweird{} edge, that $\TS$ (respectively $\TSp$) is the set of \eshort{} edges of $T$ (respectively $T'$) and that $\TW$ (respectively $\TWp$) is the set of \eweird{} edges of $T$ (respectively $T'$).
    Hence $|I_R| \leq |\TW|+|\TWp|+|\TS|+|\TSp|-2b$.
    By \cref{lem_tiny_vs_weird} and \cref{obs:minus_k}, we obtain $|\TW|+|\TWp| \leq |\TS|-2 + |\TSp| - 2 = |\TS|+|\TSp|-4$ and hence $|I_R| \leq 2|\TS|+2|\TSp|-2b-4$.
    We remark that exactly $|\TS|+|\TSp|-2b$ of the gaps in $|I_R|$ correspond to pairs with one \eshort{} edge and hence require only one flip.

    Further, recall that the set $I_N$ consists exclusively of gaps with \enormal{}-\enormal{} pairs and consider the conflict graph $H$ with vertex set $V(H)=I_N$.
    By \cref{lem:ABC-acyclic}, a maximum acyclic subset $Y$ of $H$ has size at least $\nicefrac{1}{3}\cdot|V(H)|$. 
    Let $\prec$ be a topological ordering of $H[Y]$ and let $X = I_N \setminus Y$. 

    We use one flip for each of the $|\TS|+|\TSp|-2b$ gaps in $|I_R|$ corresponding to pairs with one \eshort{} edge, two flips for all other edge pairs corresponding to gaps $X\cup I_R$, and one flip for each pair corresponding to a gap in $Y$. 
    To ease the counting, we split the total flip sequence $T \to \cdots \to T'$ into five parts:
    \begin{description}
        \item[$F_1$:] 
            ${T} \to \cdots \to {T}_{I_R}$ replaces (in any order) each edge $e_i \in {T}$ with $g_i \in I_R$ by $p_ip_{i+1}$.
        
        \item[$F_2$:] 
            ${T_{I_R}} \to \cdots \to {T}_{I_R \cup X}$ replaces (in any order) each edge $e_i \in {T}$ with $g_i \in X$ by $p_ip_{i+1}$.
	
        \item[$F_3$:] 
            ${T}_{I_R \cup X} \to \cdots \to {T}'_{I_R \cup X}$ replaces (in order according to $\prec$) each edge $e_i \in T$ with $g_i \in Y$ by the edge $e'_i \in T'$.
        
        \item[$F_4$:]
            ${T}'_{I_R \cup X} \to \cdots \to {T}'_{I_R}$ replaces (in any order) each  edge $p_ip_{i+1}$ with $g_i \in X$         by the edge $e'_i \in T'$.
            
        \item[$F_5$:] 
            ${T}'_{I_R} \to \cdots \to {T}'$ replaces (in any order) each  edge $p_ip_{i+1}$ with $g_i \in I_R$ by the edge $e'_i \in T'$.
    \end{description}
    Note that the flip sequence is valid by \cref{lem:validFlip,lem:flip-to-boundary,prop:flip-NN}.

    It remains to compute the total length of the flip sequence.
    Let $d_1=|\TS|+|\TSp|-2b$ be the number of $1$-flips for $I_R$, let $d_2=|I_R|-d_1$ be the number of $2$-flips for $I_R$ and let $d_3=|I_N|$.
    Note that $d=d_1+d_2+d_3$.

    The first and last step of the sequence require a total of $d_1+2d_2$ flips.
    The middle three steps together require $2d_3-|Y| \le \nicefrac{5}{3}\cdot d_3$ flips.

    Since $|I_R| \le |\TS|+|\TSp|-4+|\TS|+|\TSp|-2b = 2d_1+2b-4$, we have $d_2\le d_1+2b-4$.

    Altogether we obtain
    \begin{align*}
        \dist(T,T') &\leq d_1+2d_2+\frac{5}{3}d_3 
        = d_1+\frac{5}{3}d_2+\frac{1}{3}d_2 + \frac{5}{3}d_3 \\
        &\leq d_1+\frac{5}{3}d_2+\frac{1}{3}(d_1+2b)- \frac{4}{3} + \frac{5}{3}d_3 \\
        &=\frac{4}{3}d_1+\frac{5}{3}(d_2+d_3)+\frac{2}{3}b- \frac{4}{3} 
        \le \frac{5}{3}d+\frac{2}{3}b- \frac{4}{3}.
    \end{align*}

    \smallskip
    We now turn to the case $c\neq 0$.
    Consider the $c+1$ bounded cells of the convex hull of $P$ that are induced by the set $C(T,T')$; see~\cref{fig:subinstances}. 
    For each closed cell $F$ with $n(F)$ points of $P$, let $T(F)=T\cap F$ and $T'(F)=T'\cap F$ be the non-crossing spanning trees on the $n(F)$ points of $F$ induced by $T$ and $T'$, respectively, with $C(T(F),T'(F))=\emptyset$ and $b(F)=|B(T(F),T(F))|$ common boundary edges.

    Consider again the edges of $T$ and $T'$.
    Every edge of $B(T,T')$ contributes to exactly one of the $b(F)$'s and every edge of $C(T,T')$ to exactly two of them. 
    Hence $\sum_{F} b(F)=b+2c$.
    On the other hand, every edge of $D(T,T')$ (resp. $D(T',T)$) lies in exactly one cell $F$. 
    Thus, with $d(F)=|D(T(F),T'(F))|=|D(T'(F),T(F))|$, we have that $\sum_{F} d(F)=d$.
    Applying the above flip process to each of the tree pairs $(T(F),T'(F))$ independently, we obtain 
    \begin{align*}
        \dist(T,T') \leq \sum_{F}\dist(T(F),T'(F)) 
	\leq \sum_{F} \left(\frac{5}{3}d(F)+\frac{2}{3}b(F)-\frac{4}{3}\right) 
        &=  \frac{5}{3}d+\frac{2}{3}(b+2c)-\frac{4}{3}(c+1) \\
        &= \frac{5}{3}d+\frac{2}{3}b- \frac{4}{3}.
    \end{align*}

    This completes the proof of the theorem.
\end{proof}

\section{Separated Caterpillars -- Proof of \cref{thm:caterpillar}}
\label{sec:caterpillars}

In this section, we improve the upper bound for the case where one tree has a special structure, namely, if it is a \sep caterpillar.

We call a tree $T$ on a convex point set a \emph{\sep caterpillar} if it contains at most two boundary edges.

\begin{figure}[htb]
    \centering
    \includegraphics[page=1]{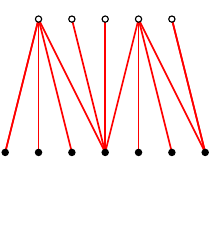}
    \caption{A \sep caterpillar with a valid 2-coloring of its points.}
    \label{fig:caterpillar}
\end{figure}

In fact, there are a few equivalent definitions.

\begin{observation}\label{lem:sepCatChar}
    Let $T$ be a tree on a convex point set (with $n\geq 3$ points) with a valid 2-coloring of the points.
    Then the following statements are equivalent.
    \begin{itemize}
	\item $T$ is a \sep caterpillar.
	\item Each color class forms a consecutive interval (along the boundary of the convex hull).
	\item The color classes can be separated by a line.
	\item $T$ contains exactly two boundary edges. 
	\item There is linear vertex labeling such that the poset defined by the edge cover relation of the edges is a total order.
        \item For every linear vertex labeling, the poset defined by the edge cover relation consists of (at most) two chains. 
        \item For every linear vertex labeling, $T$ has no \eweird edge.
    \end{itemize}
\end{observation}

We remark that Bousquet et al.~\cite{bousquet2023noteJOURNAL} have considered separated caterpillars (under the name of \emph{nice caterpillars} and with a longer definition). 
They showed that if one of two trees is a nice caterpillar, then their flip distance is at most $\nicefrac{3}{2}\cdot n$. 
Note that the lower bound examples from~\cite[Theorem 3.5]{Hernando}, illustrated in \cref{fig:IntroC}, are in fact separated caterpillars.
Thus, the bound is tight up to additive constants. 
We show that the bound even holds in terms of the difference $|T-T'|$ of the two trees $T,T'$.

\caterpillar*

Assume first that $T$ and $T'$ have at least one common chord.
As in the previous section, we can split the instance at a common chord into two subinstances. 
The common chord turns into a common boundary edge in each part, and each part of the separated caterpillar is again a separated caterpillar. 
By repeated application, we obtain a collection of subinstances without common chords.  
We will prove that~\cref{thm:caterpillar} holds for these subinstances (\cref{lem:caterpillars_without_common_chords}). 
Note that every edge in $|T-T'|$ lies in exactly one of these subinstances.
Since the bound in~\cref{thm:caterpillar} only depends on $|T-T'|$, 
applying \cref{lem:caterpillars_without_common_chords} independently to each subinstance then directly yields the desired statement for the original trees $T$ and $T'$ and hence completes the proof of~\cref{thm:caterpillar}.

For the remainder of this section, we assume that $T$ and $T'$ have no common chords.

By \cref{lem:freeGapToCut}, we can label the points $p_1,\dots,p_n$ such that neither $T$ nor $T'$ has the edge $p_1p_n$ and consider the \representation of $T$ and $T'$; otherwise, $d$ flips suffice and we are done.
It follows that both trees have at least two maximal edges, so by \cref{lem_tiny_vs_weird}, both $T$ and $T'$ have at least two more \eweird{} than \eshort{} edges.

Now we use the fact that $T$ is a separated caterpillar and thus has a special structure. 
In particular, by \cref{lem:sepCatChar}, its edges form two chains by the cover relation. 
We define $E_\ell$ as the set of all edges covering the leftmost \eshort{} edge $e_\ell$ and $E_r$ as the set of all edges covering the rightmost \eshort{} edge $e_r$; we have $e_\ell\in E_\ell$ and $e_r\in E_r$. 
Clearly, we have $T = E_\ell\cup E_r$.
For an illustration, consider \cref{fig:caterpillar2}.
Moreover, note that $T$, besides the two \eshort{} edges $e_\ell$ and $e_r$, has only \enormal{} edges.

\begin{figure}[htb]
    \centering
    \begin{subfigure}{.25\textwidth}
        \centering
        \includegraphics[page=2]{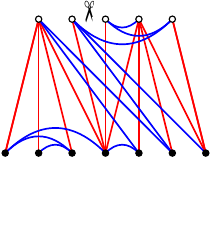}
        \caption{}
        \label{fig:caterpillar2A}
    \end{subfigure}\hfil
    \begin{subfigure}{.45\textwidth}
        \centering
        \includegraphics[page=3]{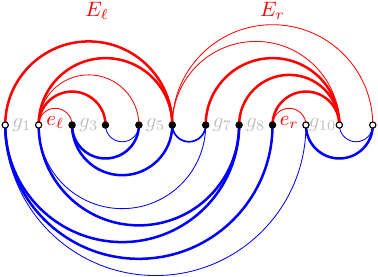}
        \caption{}
        \label{fig:caterpillar2B}
    \end{subfigure}\hfil
    \begin{subfigure}{.2\textwidth}
        \centering
        \includegraphics[page=4]{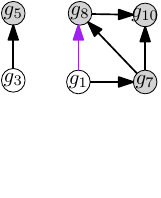}
        \caption{}
        \label{fig:caterpillar2C}
    \end{subfigure}
    \caption{Illustration for the proof of \cref{thm:caterpillar}.}
    \label{fig:caterpillar2}
\end{figure}

We pair the edges of $T$ and $T'$ via the shortest edge covering a gap as explained in \cref{sec_fundamentals} and partition the pairs into the sets $\pairs_=,\pairs_N,\pairs_R$.
For the gaps $I_N$ corresponding to \enormal{}-\enormal{} pairs, we define $A,B\subseteq I_N$ as follows:
For a gap $g$ with associated pair $(e,e')\in \pairs_N$ where $e\in E_\ell$, we let $g\in A$ if $e$ covers~$e'$, and $g\in B$ otherwise.
If $e\in E_r$, then $g\in B$ if $e$ covers $e'$, and $g\in A$ otherwise.
Clearly, $A\cup B = I_N$.

\begin{lemma}
    \label{lem_F_B_acyclic}
    $H[A]$ and $H[B]$ are acyclic.
\end{lemma}
\begin{proof}
    By left-right symmetry, it suffices to show that $H[A]$ is acyclic.
    We say that a gap $g_i\in A$ with pair $(e_i,e'_i)$ comes \emph{before} a gap $g_j$ with pair $(e_j,e_j')\in A$ if
    \begin{romanenumerate}
        \item $e_i\in E_\ell$ and $e_j\in E_r$, or\label{enum:case-1}
        \item $e_i,e_j\in E_\ell$ and $e_j$ covers $e_i$, or\label{enum:case-2}
        \item $e_i,e_j\in E_r$ and $e_i$ covers $e_j$.\label{enum:case-3}
    \end{romanenumerate}
    This gives a total order on $A$.

    For $g_i,g_j\in A$, we show that if $g_i$ comes before $g_j$, then there is no edge in the conflict graph $H$ from $g_j$ to $g_i$.
    We consider the cases \eqref{enum:case-1}--\eqref{enum:case-3} separately.
    In case \eqref{enum:case-1}, since $g_i\in A$, $e_i'$ and $e_j$ do not intersect, nor does one cover the other.
    Thus, there is no edge $\overrightarrow{g_jg_i}$ in $H$.
    In case \eqref{enum:case-2}, $e_j$ covers $e_i$, which covers $e_i'$.
    This immediately excludes type~1 and type~2 in \cref{def_conflict_graph}.
    Moreover, since $e_j$ covers $e_i$, $e_i$ does not cover $g_j$, so neither can $e'_i$, which excludes type~3.
    In case \eqref{enum:case-3}, we have that $e_i$ does not cover $e'_i$.
    Both $e_i$ and $e'_i$ cover $g_i$, which means that either (a) $e'_i$ covers $e_i$, or (b) $e_i$ and $e'_i$ cross.
    If (a), $e'_i$ covers $e_i$, which covers $e_j$, so types~1 and 3 of \cref{def_conflict_graph} are excluded.
    Since $e_i$ covers $e_j$, $e_j$ does not cover $g_i$, so also type~2 is excluded.
    If (b), then since $e_i$ and $e'_i$ are \enormal{}, the only gap covered by both is $g_i$.
    This gap is not covered by $e_j$, so there is no gap covered by both $e'_i$ and $e_j$.
    It follows that there is no edge $\overrightarrow{g_jg_i}$ in $H$ in either of the cases \eqref{enum:case-1}--\eqref{enum:case-3} and thus $H[A]$ is acyclic.
\end{proof}

We now have all tools to present the flip sequence.

\begin{lemma}\label{lem:caterpillars_without_common_chords}
    There exists a flip sequence $F$ from $T$ to $T'$ of length at most $\nicefrac{3}{2}\cdot d$.
\end{lemma}
\begin{proof}
    We start by describing our flip sequence $F$ which consists of four parts. 
    Choose $Y$ as a largest acyclic subset of $H$ among $A$ and $B$ and let $X=I_N-Y$. 
    Recall that  $|T'\cap T|=|\{e_\ell,e_r\}|=2$.
    Hence, we have $d=n-3$.

    \begin{description}
        \item[$F_1$:]  
            For each $(e_i,e_i')\in\pairs_R$  with gap $g_i$ where $e'$ is \eshort{} or \eweird{}, flip $e$ to $p_ip_{i+1}$.
            Clearly, $|F_1|=|\TSp|+|\TWp|-2$; recall that $T'$ contains the two \eshort{} edges $e_\ell,e_r$ which belong to pairs in $\pairs_=$.

        \item [$F_2$:]
            For each $g_i\in X$, let $(e_i,e'_i)\in\pairs_N$ denote the corresponding pair.
            We flip $e$ to $p_ip_{i+1}$.
            Clearly, $|F_2|=2|X|$.

        \item[$F_3$:]  
            For each $g_i\in Y$, let $(e_i,e'_i)\in\pairs_N$ denote the corresponding pair.
            We flip $e_i$ to $e'_i$.
            Clearly, $|F_3|=|Y|$.
 
        \item[$F_4$:]
            For each $e'\in \TWp$ with corresponding gap $g_k$, perform flip $p_kp_{k+1}\to e'$.
            Clearly, $|F_4|=|\TWp|$.
    \end{description}

    The validity of the flip sequences in $F_1,F_2,$ and $F_4$ follow from \cref{,lem:flip-to-boundary} as we introduce a boundary edge or remove a boundary edge covering the same gap as its partner edge. 
    Let us denote by $T_1$ the tree resulting from applying $F_1$ and $F_2$ to $T$ and by $T_2$ the tree resulting from applying the inverse of $F_4$ to $T'$; note that all these flips can be applied in any order.
    Then, for $T_1$ and $T_2$, $\pairs_R=\emptyset$, $\pairs_N$ corresponds to $Y$, and $H(T_1,T_2)[Y]$ is acyclic. 
    Hence, \cref{prop:flip-NN} guarantees a flip sequence of length $|Y|$.

    It remains to discuss the total length.
    By \cref{lem_tiny_vs_weird} and the fact that $T'$ has at least two uncovered edges, we have $|\TWp|\leq |\TSp|-2$.
    Hence,
    \[
        |F_1|+|F_4|=|\TSp|+2|\TWp|-c\leq \nicefrac 3 2\cdot(|\TSp|+|\TWp|-2).
    \]

    By \cref{lem_F_B_acyclic}, we have $|Y|\geq \nicefrac{1}{2}\cdot|I_N|=\nicefrac{1}{2}\cdot|\TNp|$ and thus $|F_2|+|F_3|=2|X|+|Y|=\nicefrac{3}{2}\cdot|\TNp|$.
    Therefore, we obtain the following bound
    \begin{align*}
        |F|&= |F_1| + |F_2| + |F_3| + |F_4| \leq 
        \nicefrac 3 2\cdot(|\TSp|+|\TWp| +|\TNp| -2)= \nicefrac 3 2\cdot(n-3)=\nicefrac{3}{2}\cdot d,
    \end{align*}
    which concludes the proof.
\end{proof}

\section{Discussion and Open Problems}
\label{sec:conclusions}

In this work, we improved the lower and upper bounds  on the diameter of $\mathcal F_n$.
Together, \cref{thm:careful-main-upper-bound,thm:main-lower-bound} yield 
\[
    \nicefrac{14}{9}\cdot n - \mathcal{O}(1) \leq \diam(\mathcal{F}_n) \leq \nicefrac{5}{3} \cdot n - 3 = \nicefrac{15}{9}\cdot n - 3.
\]
Thus, the gap between the upper and lower bounds on $\diam(\mathcal{F}_n)$ has been tightened from about $0.45\cdot n$ to just $\nicefrac 19 \cdot n + \mathcal{O}(1)$.
With \cref{thm:alpha} at hand, closing the gap can be achieved by improving the lower bound $\frac{\las(H)}{|V(H)|} \geq \nicefrac 13$ for all conflict graphs $H$, or by presenting a conflict graph $H$ with $\frac{\las(H)}{|V(H)|} < \nicefrac{4}{9}$.
We therefore believe that our techniques have potential to help determining $\diam(\mathcal{F}_n)$ completely.

Let us note that the new lower bound of $\nicefrac{14}{9}\cdot n$ for the convex setting actually improves upon the best known lower bounds not only for points in general position, but also for more restricted flip operations, e.g., the \emph{compatible edge exchange} (where the exchanged edges are non-crossing), the \emph{rotation} (where the exchanged edge are adjacent), and the \emph{edge slide} (where the exchanged edges together with some third edge form an uncrossed triangle).
For an overview of best known bounds for five studied flip types, we refer to Nichols et al.~\cite{TreeTransition}. 

We also considered bounding the flip distance $\dist(T,T')$ of two trees $T,T'$ in terms of $d = |T,T'|$ and obtained an upper bound of $\dist(T,T')\leq \nicefrac 53 \cdot d+\nicefrac 23 \cdot b-\nicefrac 43$, where $b$ is the number of common boundary edges of $T$ and $T'$.
We can interpret \cref{thm:careful-main-upper-bound} as being halfway between an upper bound on $\dist(T,T')$ in terms of $n$ and one in terms of $d$:
common chords do not contribute at all, and common boundary edges contribute less than the edges in the symmetric difference.
If $\nicefrac{2}{3}\cdot b-\nicefrac{4}{3}$ can be removed from the bound, then this would give a tight upper bound in terms of $d$.
In fact, Bousquet et al.~\cite[Theorem~4]{bousquet2024reconfigurationSoCG}, presented graphs $T_d$ and $T'_d$ with $\dist(T_d,T'_d) = \nicefrac{5}{3}\cdot d$ (for all $d$ divisible by~$3$).    

Besides determining the maximum flip distance in terms of $n$ or $d$ for the mentioned settings, it is also interesting to investigate the computational complexity of computing a shortest flip sequence for two given non-crossing trees.
Is it \textsc{NP}-complete or polynomial-time solvable?
The question is open for both settings of convex and general position.

Moreover, Aichholzer et al. conjecture that for any two trees $T,T'$ there exists a flip sequence of length $\dist(T,T')$ such that common edges (so-called \emph{happy edges}) are not flipped~\cite[Conjecture 16]{aichholzer2022reconfiguration}, and that there exists a flip sequence of length $\dist(T,T')$ such that any edge appearing in the sequence is either in $T\cup T'$ or a boundary edge~\cite[Conjecture 20]{aichholzer2022reconfiguration}.
We note that the flip sequences we construct for \cref{thm:careful-main-upper-bound}, which is our strongest upper bound, have both these properties, providing some evidence in favor of the conjectures.

\bibliographystyle{plainurl}
\bibliography{lit}

\begin{thebibliography}{10}

\bibitem{AichholzerAH02}
Oswin Aichholzer, Franz Aurenhammer, and Ferran Hurtado.
\newblock Sequences of spanning trees and a fixed tree theorem.
\newblock {\em Computational Geometry: Theory \& Applications}, 21(1-2):3--20,
  2002.
\newblock \href {https://doi.org/10.1016/S0925-7721(01)00042-6}
  {\path{doi:10.1016/S0925-7721(01)00042-6}}.

\bibitem{aichholzer2022reconfiguration}
Oswin Aichholzer, Brad Ballinger, Therese Biedl, Mirela Damian, Erik~D.
  Demaine, Matias Korman, Anna Lubiw, Jayson Lynch, Josef Tkadlec, and Yushi
  Uno.
\newblock Reconfiguration of non-crossing spanning trees, 2022.
\newblock \href {https://doi.org/10.48550/arXiv.2206.03879}
  {\path{doi:10.48550/arXiv.2206.03879}}.

\bibitem{2023Aicholzer}
Oswin Aichholzer, Kristin Knorr, Wolfgang Mulzer, Johannes Obenaus, Rosna Paul,
  and Birgit Vogtenhuber.
\newblock Flipping plane spanning paths.
\newblock In {\em International Conference and Workshops on Algorithms and
  Computation (WALCOM)}, pages 49--60, 2023.
\newblock \href {https://doi.org/10.1007/978-3-031-27051-2_5}
  {\path{doi:10.1007/978-3-031-27051-2_5}}.

\bibitem{akl2007planar}
Selim~G. Akl, Md~Kamrul Islam, and Henk Meijer.
\newblock On planar path transformation.
\newblock {\em Information Processing Letters}, 104(2):59--64, 2007.
\newblock \href {https://doi.org/10.1016/j.ipl.2007.05.009}
  {\path{doi:10.1016/j.ipl.2007.05.009}}.

\bibitem{AvisFukuda}
David Avis and Komei Fukuda.
\newblock Reverse search for enumeration.
\newblock {\em Discrete Applied Mathematics}, 65(1-3):21--46, 1996.
\newblock \href {https://doi.org/10.1016/0166-218X(95)00026-N}
  {\path{doi:10.1016/0166-218X(95)00026-N}}.

\bibitem{bose2009flipsinplanar}
Prosenjit Bose and Ferran Hurtado.
\newblock Flips in planar graphs.
\newblock {\em Computational Geometry: Theory \& Applications}, 42(1):60--80,
  2009.
\newblock \href {https://doi.org/10.1016/j.comgeo.2008.04.001}
  {\path{doi:10.1016/j.comgeo.2008.04.001}}.

\bibitem{bousquet2024reconfigurationSoCG}
Nicolas Bousquet, Lucas de~Meyer, Th\'{e}o Pierron, and Alexandra Wesolek.
\newblock Reconfiguration of plane trees in convex geometric graphs.
\newblock In {\em International Symposium on Computational Geometry (SoCG
  2024)}, volume 293 of {\em LIPIcs}, pages 22:1--22:17, 2024.
\newblock \href {http://arxiv.org/abs/2310.18518} {\path{arXiv:2310.18518}},
  \href {https://doi.org/10.4230/LIPIcs.SoCG.2024.22}
  {\path{doi:10.4230/LIPIcs.SoCG.2024.22}}.

\bibitem{bousquet2023noteJOURNAL}
Nicolas Bousquet, Valentin Gledel, Jonathan Narboni, and Th{\'{e}}o Pierron.
\newblock A note on the flip distance between non-crossing spanning trees.
\newblock {\em Computing in Geometry and Topology}, 2(1):8:1--8:7, 2023.
\newblock \href {http://arxiv.org/abs/2303.07710} {\path{arXiv:2303.07710}},
  \href {https://doi.org/10.57717/cgt.v2i1.36}
  {\path{doi:10.57717/cgt.v2i1.36}}.

\bibitem{bousquet_et_al:LIPIcs.ESA.2020.24}
Nicolas Bousquet, Takehiro Ito, Yusuke Kobayashi, Haruka Mizuta, Paul Ouvrard,
  Akira Suzuki, and Kunihiro Wasa.
\newblock Reconfiguration of spanning trees with many or few leaves.
\newblock In {\em European Symposium on Algorithms (ESA 2020)}, volume 173 of
  {\em LIPIcs}, pages 24:1--24:15, 2020.
\newblock \href {https://doi.org/10.4230/LIPIcs.ESA.2020.24}
  {\path{doi:10.4230/LIPIcs.ESA.2020.24}}.

\bibitem{bousquetTreesDegree}
Nicolas Bousquet, Takehiro Ito, Yusuke Kobayashi, Haruka Mizuta, Paul Ouvrard,
  Akira Suzuki, and Kunihiro Wasa.
\newblock Reconfiguration of spanning trees with degree constraints or diameter
  constraints.
\newblock {\em Algorithmica}, 85(9):2779--2816, 2023.
\newblock \href {https://doi.org/10.1007/s00453-023-01117-z}
  {\path{doi:10.1007/s00453-023-01117-z}}.

\bibitem{CSZ15}
Cesar Ceballos, Francisco Santos, and G{\"u}nter~M. Ziegler.
\newblock Many non-equivalent realizations of the associahedron.
\newblock {\em Combinatorica}, 35(5):513--551, 2015.
\newblock \href {https://doi.org/10.1007/s00493-014-2959-9}
  {\path{doi:10.1007/s00493-014-2959-9}}.

\bibitem{convexDiameter}
Jou-Ming Chang and Ro-Yu Wu.
\newblock On the diameter of geometric path graphs of points in convex
  position.
\newblock {\em Information Processing Letters}, 109(8):409--413, 2009.
\newblock \href {https://doi.org/10.1016/j.ipl.2008.12.017}
  {\path{doi:10.1016/j.ipl.2008.12.017}}.

\bibitem{Eppstein10_socg07}
David Eppstein.
\newblock Happy endings for flip graphs.
\newblock In Jeff Erickson, editor, {\em Proceedings of the 23rd {ACM}
  Symposium on Computational Geometry, Gyeongju, South Korea, June 6-8, 2007},
  pages 92--101. {ACM}, 2007.
\newblock \href {https://doi.org/10.1145/1247069.1247084}
  {\path{doi:10.1145/1247069.1247084}}.

\bibitem{Eppstein10}
David Eppstein.
\newblock Happy endings for flip graphs.
\newblock {\em Journal of Computational Geometry}, 1(1):3--28, 2010.
\newblock \href {https://doi.org/10.20382/JOCG.V1I1A2}
  {\path{doi:10.20382/JOCG.V1I1A2}}.

\bibitem{EF23}
David Eppstein and Daniel Frishberg.
\newblock Improved mixing for the convex polygon triangulation flip walk.
\newblock In {\em International Colloquium on Automata, Languages, and
  Programming (ICALP 2023)}, volume 261 of {\em LIPIcs}, pages 56:1--56:17,
  2023.
\newblock \href {https://doi.org/10.4230/LIPIcs.ICALP.2023.56}
  {\path{doi:10.4230/LIPIcs.ICALP.2023.56}}.

\bibitem{rainbowSoCG18}
Stefan Felsner, Linda Kleist, Torsten M{\"{u}}tze, and Leon Sering.
\newblock Rainbow cycles in flip graphs.
\newblock In Bettina Speckmann and Csaba~D. T{\'{o}}th, editors, {\em 34th
  International Symposium on Computational Geometry, SoCG 2018, June 11-14,
  2018, Budapest, Hungary}, volume~99 of {\em LIPIcs}, pages 38:1--38:14.
  Schloss Dagstuhl - Leibniz-Zentrum f{\"{u}}r Informatik, 2018.
\newblock URL: \url{https://doi.org/10.4230/LIPIcs.SoCG.2018.38}, \href
  {https://doi.org/10.4230/LIPICS.SOCG.2018.38}
  {\path{doi:10.4230/LIPICS.SOCG.2018.38}}.

\bibitem{rainbowJournal}
Stefan Felsner, Linda Kleist, Torsten M{\"u}tze, and Leon Sering.
\newblock Rainbow cycles in flip graphs.
\newblock {\em SIAM Journal on Discrete Mathematics}, 34(1):1--39, 2020.
\newblock \href {https://doi.org/10.1137/18M1216456}
  {\path{doi:10.1137/18M1216456}}.

\bibitem{HernandoHH02}
Carmen Hernando, Michael~E. Houle, and Ferran Hurtado.
\newblock On local transformation of polygons with visibility properties.
\newblock {\em Theoretical Computer Science}, 289(2):919--937, 2002.
\newblock \href {https://doi.org/10.1016/S0304-3975(01)00409-1}
  {\path{doi:10.1016/S0304-3975(01)00409-1}}.

\bibitem{Hernando}
Carmen Hernando, Ferran Hurtado, Alberto M{\'{a}}rquez, Merc{\`{e}} Mora, and
  Marc Noy.
\newblock Geometric tree graphs of points in convex position.
\newblock {\em Discrete Applied Mathematics}, 93(1):51--66, 1999.
\newblock \href {https://doi.org/10.1016/S0166-218X(99)00006-2}
  {\path{doi:10.1016/S0166-218X(99)00006-2}}.

\bibitem{hernandoTreeLabeled}
Carmen Hernando, Ferran Hurtado, Merc{\`e} Mora, and Eduardo Rivera-Campo.
\newblock Grafos de {\'a}rboles etiquetados y grafos de {\'a}rboles
  geom{\'e}tricos etiquetados.
\newblock {\em Proc. X Encuentros de Geometra Computacional}, pages 13--19,
  2003.

\bibitem{perfect-matchings}
Carmen Hernando, Ferran Hurtado, and Marc Noy.
\newblock Graphs of non-crossing perfect matchings.
\newblock {\em Graphs and Combinatorics}, 18(3):517--532, 2002.
\newblock \href {https://doi.org/10.1007/S003730200038}
  {\path{doi:10.1007/S003730200038}}.

\bibitem{HouleHNR05}
Michael~E. Houle, Ferran Hurtado, Marc Noy, and Eduardo Rivera{-}Campo.
\newblock Graphs of triangulations and perfect matchings.
\newblock {\em Graphs and Combinatorics}, 21(3):325--331, 2005.
\newblock \href {https://doi.org/10.1007/S00373-005-0615-2}
  {\path{doi:10.1007/S00373-005-0615-2}}.

\bibitem{HurtadoNU99}
Ferran Hurtado, Marc Noy, and Jorge Urrutia.
\newblock Flipping edges in triangulations.
\newblock {\em Discrete \& Computational Geometry}, 22(3):333--346, 1999.
\newblock \href {https://doi.org/10.1007/PL00009464}
  {\path{doi:10.1007/PL00009464}}.

\bibitem{graphAssoHardness}
Takehiro Ito, Naonori Kakimura, Naoyuki Kamiyama, Yusuke Kobayashi, Shun-ichi
  Maezawa, Yuta Nozaki, and Yoshio Okamoto.
\newblock Hardness of finding combinatorial shortest paths on graph
  associahedra.
\newblock In {\em International Colloquium on Automata, Languages, and
  Programming (ICALP 2023)}, volume 261 of {\em LIPIcs}, pages 82:1--82:17,
  2023.
\newblock \href {https://doi.org/10.4230/LIPIcs.ICALP.2023.82}
  {\path{doi:10.4230/LIPIcs.ICALP.2023.82}}.

\bibitem{KKR_paths}
Linda Kleist, Peter Kramer, and Christian Rieck.
\newblock On the connectivity of the flip graph of plane spanning paths.
\newblock In {\em International Workshop on Graph-Theoretic Concepts in
  Computer Science}, 2024.

\bibitem{Lawson72}
Charles~L. Lawson.
\newblock Transforming triangulations.
\newblock {\em Discrete Mathematics}, 3(4):365--372, 1972.
\newblock \href {https://doi.org/10.1016/0012-365X(72)90093-3}
  {\path{doi:10.1016/0012-365X(72)90093-3}}.

\bibitem{LubiwP15}
Anna Lubiw and Vinayak Pathak.
\newblock Flip distance between two triangulations of a point set is
  {NP}-complete.
\newblock {\em Computational Geometry: Theory \& Applications}, 49:17--23,
  2015.
\newblock \href {https://doi.org/10.1016/J.COMGEO.2014.11.001}
  {\path{doi:10.1016/J.COMGEO.2014.11.001}}.

\bibitem{lucas1987}
Joan~M Lucas.
\newblock The rotation graph of binary trees is hamiltonian.
\newblock {\em Journal of Algorithms}, 8(4):503--535, 1987.
\newblock \href {https://doi.org/10.1016/0196-6774(87)90048-4}
  {\path{doi:10.1016/0196-6774(87)90048-4}}.

\bibitem{MilichMP21}
Marcel Milich, Torsten M{\"{u}}tze, and Martin Pergel.
\newblock On flips in planar matchings.
\newblock {\em Discrete Applied Mathematics}, 289:427--445, 2021.
\newblock \href {https://doi.org/10.1016/J.DAM.2020.10.018}
  {\path{doi:10.1016/J.DAM.2020.10.018}}.

\bibitem{chromatic}
Ruy~Fabila Monroy, David Flores{-}Pe{\~{n}}aloza, Clemens Huemer, Ferran
  Hurtado, David~R. Wood, and Jorge Urrutia.
\newblock On the chromatic number of some flip graphs.
\newblock {\em Discrete Mathematics \& Theoretical Computer Science},
  11(2):47--56, 2009.
\newblock \href {https://doi.org/10.46298/DMTCS.460}
  {\path{doi:10.46298/DMTCS.460}}.

\bibitem{TreeTransition}
Torrie~L Nichols, Alexander Pilz, Csaba~D T{\'o}th, and Ahad~N Zehmakan.
\newblock Transition operations over plane trees.
\newblock {\em Discrete Mathematics}, 343(8):111929, 2020.
\newblock \href {https://doi.org/10.1016/j.disc.2020.111929}
  {\path{doi:10.1016/j.disc.2020.111929}}.

\bibitem{nishimuraIntroReconfiguration}
Naomi Nishimura.
\newblock Introduction to reconfiguration.
\newblock {\em Algorithms}, 11(4):52, 2018.
\newblock \href {https://doi.org/10.3390/a11040052}
  {\path{doi:10.3390/a11040052}}.

\bibitem{pilz2014flip}
Alexander Pilz.
\newblock Flip distance between triangulations of a planar point set is
  {APX}-hard.
\newblock {\em Computational Geometry}, 47(5):589--604, 2014.
\newblock \href {https://doi.org/10.1016/j.comgeo.2014.01.001}
  {\path{doi:10.1016/j.comgeo.2014.01.001}}.

\bibitem{pournin2014diameter}
Lionel Pournin.
\newblock The diameter of associahedra.
\newblock {\em Advances in Mathematics}, 259:13--42, 2014.
\newblock \href {https://doi.org/10.1016/j.aim.2014.02.035}
  {\path{doi:10.1016/j.aim.2014.02.035}}.

\bibitem{hamilton}
Eduardo Rivera{-}Campo and Virginia Urrutia{-}Galicia.
\newblock Hamilton cycles in the path graph of a set of points in convex
  position.
\newblock {\em Computational Geometry: Theory \& Applications}, 18(2):65--72,
  2001.
\newblock \href {https://doi.org/10.1016/S0925-7721(00)00026-2}
  {\path{doi:10.1016/S0925-7721(00)00026-2}}.

\bibitem{Rokicki2010Rubik}
Tomas Rokicki, Herbert Kociemba, Morley Davidson, and John Dethridge.
\newblock God's number is 20.
\newblock last accessed: July 2024.
\newblock URL: \url{http://www.cube20.org}.

\bibitem{sleator1986rotation}
Daniel~Dominic Sleator, Robert~Endre Tarjan, and William~P. Thurston.
\newblock Rotation distance, triangulations, and hyperbolic geometry.
\newblock In {\em Symposium on Theory of Computing (STOC)}, pages 122--135,
  1986.
\newblock \href {https://doi.org/10.1145/12130.12143}
  {\path{doi:10.1145/12130.12143}}.

\bibitem{heuvel2003survey}
Jan van~den Heuvel.
\newblock The complexity of change.
\newblock {\em Surveys in Combinatorics 2013}, pages 127--160, 2013.
\newblock \href {https://doi.org/10.1017/cbo9781139506748.005}
  {\path{doi:10.1017/cbo9781139506748.005}}.

\bibitem{WagnerW22_soda20}
Uli Wagner and Emo Welzl.
\newblock Connectivity of triangulation flip graphs in the plane (part {I:}
  edge flips).
\newblock In Shuchi Chawla, editor, {\em Proceedings of the 2020 {ACM-SIAM}
  Symposium on Discrete Algorithms, {SODA} 2020, Salt Lake City, UT, USA,
  January 5-8, 2020}, pages 2823--2841. {SIAM}, 2020.
\newblock \href {https://doi.org/10.1137/1.9781611975994.172}
  {\path{doi:10.1137/1.9781611975994.172}}.

\bibitem{WagnerW22_socg20}
Uli Wagner and Emo Welzl.
\newblock Connectivity of triangulation flip graphs in the plane (part {II:}
  bistellar flips).
\newblock In Sergio Cabello and Danny~Z. Chen, editors, {\em 36th International
  Symposium on Computational Geometry, SoCG 2020, June 23-26, 2020,
  Z{\"{u}}rich, Switzerland}, volume 164 of {\em LIPIcs}, pages 67:1--67:16.
  Schloss Dagstuhl - Leibniz-Zentrum f{\"{u}}r Informatik, 2020.
\newblock URL: \url{https://doi.org/10.4230/LIPIcs.SoCG.2020.67}, \href
  {https://doi.org/10.4230/LIPICS.SOCG.2020.67}
  {\path{doi:10.4230/LIPICS.SOCG.2020.67}}.

\bibitem{WagnerW22}
Uli Wagner and Emo Welzl.
\newblock Connectivity of triangulation flip graphs in the plane.
\newblock {\em Discrete \& Computational Geometry}, 68(4):1227--1284, 2022.
\newblock \href {https://doi.org/10.1007/S00454-022-00436-2}
  {\path{doi:10.1007/S00454-022-00436-2}}.

\end{thebibliography}

\end{document}